\documentclass[manuscript]{acmart}
\usepackage{amsmath}
\usepackage{algorithmic}
\usepackage{algorithm2e}
\usepackage{graphicx}
\usepackage{textcomp}
\usepackage{xcolor}
\usepackage{listings}
\usepackage{paralist}
\usepackage{multirow}
\newtheorem{lemma}{Lemma}

\newtheorem{example}{Example}

\newsavebox{\choleskyboxA}
\newsavebox{\choleskyboxB}
\newsavebox{\choleskyboxC}
\newsavebox{\mvbox}
\newsavebox{\spmspvboxA}
\newsavebox{\spmspvboxB}
\newsavebox{\spmvcsrbox}
\newsavebox{\spmsvcsrbox}
\newsavebox{\spmvdiagbox}
\newsavebox{\spmsvdiagbox}
\def\BibTeX{{\rm B\kern-.05em{\sc i\kern-.025em b}\kern-.08em
    T\kern-.1667em\lower.7ex\hbox{E}\kern-.125emX}}
\begin{document}

\begin{lrbox}{\choleskyboxA}
	\begin{lstlisting}[mathescape]
${\tt for(i\texttt{=}0;i\texttt{<}n;i\texttt{++})\{}$
 ${\tt for(j\texttt{=}0;j\texttt{<}i;j\texttt{++})\{}$
  ${\tt for(k\texttt{=}0;k\texttt{<}j;k\texttt{++})\{}$
   ${\tt S_1:A[i][j]\texttt{-=}A[i][k]\times A[j][k];}$
  ${\tt \}}$
  ${\tt if(A[j][j]\texttt{!=}0)}$
   ${\tt S_2:A[i][j]\texttt{/=}A[j][j];}$
 ${\tt \}}$
 ${\tt for(l\texttt{=}0;l\texttt{<}i;l\texttt{++})\{}$
  ${\tt S_3:A[i][i]\texttt{-=}A[i][l]\times A[i][l];}$
 ${\tt\}}$
 ${\tt S_4:A[i][i]\texttt{=}sqrt(A[i][i]);}$
${\tt\}}$
    \end{lstlisting}
\end{lrbox}
\begin{lrbox}{\choleskyboxB}
\begin{lstlisting}[mathescape]
${\tt \texttt {for(int }i=0; i<3; i++)}$
$\ \ \ \ {\tt S_4:A_{val}[2*i+0]\texttt{=}sqrt(A_{val}[2*i+0]);}$
${\tt \texttt{if(}A_{val}[4]\texttt{!=}0)}$
$\ \ \ \ {\tt S_2:A_{val}[6]\texttt{/=}A_{val}[4];}$
${\tt S_3:A_{val}[7]\texttt{-=}A_{val}[6]\times A_{val}[6];}$
${\tt S_4:A_{val}[7]\texttt{=}sqrt(A_{val}[7]);}$
${\tt \texttt{if(}A_{val}[7]\texttt{!=}0)}$
$\ \ \ \ {\tt S_2:A_{val}[10]\texttt{/=}A_{val}[7];}$
${\tt S_3:A_{val}[11]\texttt{-=}A_{val}[10]\times A_{val}[10];}$
${\tt \hdots}$
\end{lstlisting} 
\end{lrbox}
\begin{lrbox}{\choleskyboxC}
\begin{lstlisting}[mathescape]
${\tt \texttt {for(int i=}0; i<15439; i++)}\{$
$\ \ \ \ {\tt S_4:valA[i]\texttt{=}sqrt(valA[i]);}$
$\}$
\end{lstlisting} 
\end{lrbox}
\begin{lrbox}{\mvbox}
\begin{lstlisting}[mathescape]
${\tt for(i\texttt{=}0;i\texttt{<}n;i\texttt{++})\{}$
 ${\tt for(j\texttt{=}0;j\texttt{<}n;j\texttt{++})}$
  ${\tt S:Y[i]\texttt{+=}A[i][j]\times X[j];}$
$\}$ 
\end{lstlisting}
\end{lrbox}
\begin{lrbox}{\spmvcsrbox}
\begin{lstlisting}[mathescape]
${\tt for(i\texttt{=}0;i\texttt{<}n;i\texttt{++})\{}$
 ${\tt for(j\texttt{=}ptr[i];j\texttt{<}ptr[i+1];j\texttt{++})}$
  ${\tt y[i]\texttt{+=}A_{val}[j]\times X[A_{col}[j]];}$
$\}$ 
\end{lstlisting}
\end{lrbox}
\begin{lrbox}{\spmsvcsrbox}
\begin{lstlisting}[mathescape]
${\tt for(i\texttt{=}0;i\texttt{<}5;i\texttt{++})\{}$
 ${\tt for(j\texttt{=}ptr[i];j\texttt{<}ptr[i+1];j\texttt{++})\{}$
  ${\tt if(X[A_{col}[j]] \texttt{!=} 0)}$
   ${\tt Y[i]\texttt{+=}A_{val}[j]\times X[A_{col}[j]];}$
$\}\}$
\end{lstlisting}
\end{lrbox}
\begin{lrbox}{\spmvdiagbox}
\begin{lstlisting}[mathescape]
${\tt for(i\texttt{=}0;i\texttt{<}n;i\texttt{++})\{}$
 ${\tt y[i]\texttt{=}A_{val}[i]\times X[i];}$
$\}$ 
\end{lstlisting}
\end{lrbox}
\begin{lrbox}{\spmsvdiagbox}
\begin{lstlisting}[mathescape]
${\tt for(i\texttt{=}0;i\le 2;i\texttt{++})\{}$
 ${\tt Y_{val}[i]\texttt{+=}A_{val}[i+2]\times X_{val}[0];}$
$\}$ 
${\tt Y_{val}[2]\texttt{+=}A_{val}[5]\times X_{val}[1];}$
\end{lstlisting}
\end{lrbox}

\title{SpComp: A Sparsity Structure-Specific Compilation of Matrix Operations}

\author{Barnali Basak}
\affiliation{
  \institution{Indian Institute of Technology Bombay}
  \country{India}}
\email{bbasak@cse.iitb.ac.in}
\author{Uday P. Khedker}
\affiliation{
  \institution{Indian Institute of Technology Bombay}
  \country{India}}
\email{uday@cse.iitb.ac.in}
\author{Supratim Biswas}
\affiliation{
  \institution{Indian Institute of Technology Bombay}
  \country{India}}
\email{sb@cse.iitb.ac.in}

\begin{abstract}
Sparse matrix operations involve a large number of zero operands which makes most of the operations redundant. 
The amount of redundancy magnifies when a matrix operation repeatedly executes on sparse data. Optimizing matrix 
operations for sparsity involves either reorganization of data or reorganization of computations, performed either 
at compile-time or run-time. Although compile-time techniques avert from introducing run-time overhead, their application either 
is limited to simple sparse matrix operations generating dense output and handling immutable sparse matrices or 
requires manual intervention to customize the technique to different matrix operations.  

We contribute a sparsity structure-specific compilation technique, called {\em SpComp}, that optimizes a sparse 
matrix operation by automatically customizing its computations to the positions of non-zero values of the data. 
Our approach neither incurs any run-time overhead nor requires any manual intervention. It is also applicable to complex matrix 
operations generating sparse output and handling mutable sparse matrices. We introduce a data-flow analysis, 
named \textit{Essential Indices Analysis}, that statically collects the symbolic information about the computations 
and helps the code generator to reorganize the computations. The generated code includes piecewise-regular loops, 
free from indirect references and amenable to further optimization.

We see a substantial performance gain by SpComp-generated Sparse Matrix-Sparse Vector Multiplication (SpMSpV) 
code when compared against the state-of-the-art {\em TACO} compiler and piecewise-regular code generator. On average, we achieve 
${\tt\approx 79\%}$ performance gain against TACO and ${\tt\approx 83\%}$ performance gain against 
the piecewise-regular code generator. When compared against the {\em CHOLMOD} library, SpComp generated 
sparse Cholesky decomposition code showcases ${\tt \approx 65\%}$ performance gain on average.
\end{abstract}

\maketitle

\section{Introduction}
Sparse matrix operations are ubiquitous in computational science areas like circuit simulation, power dynamics, image
processing, structure modeling, data science, etc. The presence of a significant amount of zero values in the sparse 
matrices makes a considerable amount of computations, involved in the matrix operation, redundant. Only the computations 
computing non-zero values remain useful or non-redundant.

In simulation-like scenarios, a matrix operation repeatedly executes on sparse matrices whose positions or indices 
of non-zero values, better known as sparsity structures, remain unchanged although the values in these positions may change. 
For example, Cholesky decomposition in circuit simulation repeatedly decomposes the input matrix in each iteration until 
the simulation converges. The input matrix represents the physical connections of the underlying circuit and thus, 
remains unchanged throughout the simulation, although the values associated with the connections may vary. In such a scenario, 
the redundant computations caused by computing zero values get multi-fold and therefore, it is prudent to claim substantial 
performance benefits by altering the execution based on the fixed sparsity structure.

\begin{figure}[htbp]
	\centering
	\scalebox{.85}{
	\begin{tabular}{ccc}
		\begin{tabular}{cc}
			\includegraphics[height=2.5cm,width=3.3cm]{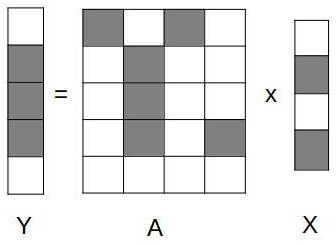}	 
			& 
			\includegraphics[height=2.5cm,width=2.8cm]{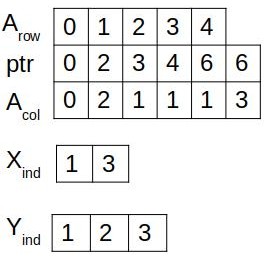}	
		\end{tabular}
		&
		\begin{tabular}{c}
			\usebox{\spmsvcsrbox}
		\end{tabular}
		&
		\begin{tabular}{c}
			\usebox{\spmsvdiagbox}
		\end{tabular} \\
		(a) & (b) & (c)
	\end{tabular}}
	\caption{(a) Sparse matrix operation SpMSpV; multiplication of sparse matrix ${\tt A}$ 
	of size ${\tt 5\times 4}$ and sparse vector ${\tt X}$ of size ${\tt 4\times 1}$, resulting 
	sparse vector ${\tt Y}$ of size ${\tt 5\times 1}$. Colored boxes denote non-zero elements. 
	Sparse matrix ${\tt A}$ stored in CSR format, sparse vectors ${\tt X}$ and ${\tt Y}$ stored 
	in COO format. ${\tt A_{val}}$, ${\tt X_{val}}$, ${\tt Y_{val}}$ hold non-zero values, not depicted here. 
	(b) SpMSpV operation operating on reorganized sparse matrix ${\tt A}$, stored using CSR format.
	(c) Reorganized computations of SpMSpV, customized to the sparsity structure of the output sparse vector ${\tt Y}$.}
	\label{fig:motiv}
\end{figure}

The state-of-the-art on avoiding redundant computations categorically employs either  
\begin{inparaenum}[(a)]
	\item reorganizing sparse data in a storage format such that the generic 
	computations operate only on the non-zero data or 
	\item reorganizing computations to restrict them to non-zero data without requiring any specific reorganization of the data.
\end{inparaenum}
Figure~\ref{fig:motiv} demonstrates the avoidance of redundant computations of Sparse Matrix-Sparse Vector 
Multiplication operation (SpMSpV) by reorganizing data and reorganizing computations. The operation multiplies 
sparse matrix ${\tt A}$ to sparse vector ${\tt X}$ and stores the result in the output sparse vector ${\tt Y}$. 
Figure~\ref{fig:motiv}(b) presents the computations on reorganized matrix ${\tt A}$ stored using Compressed 
Sparse Row (CSR) format. As a result, the non-zero values are accessed using indirect reference 
${\tt X[A_{col}[j]]}$. On the contrary, Figure~\ref{fig:motiv}(c) presents the reorganized SpMSpV computations 
customized to the positions of the non-zero elements of ${\tt Y}$. Clearly, reorganized computations result in 
direct references and a minimum number of computations. 

Approaches reorganizing computations avoid redundant computations by generating sparsity structure-specific execution. 
This can be done either at run-time or at compile-time. Run-time techniques like \textit{inspection-execution}~\cite{ins_exe} 
exploit the memory traces at run-time. The {\em executor} executes the optimized schedule generated by the {\em inspector} 
after analyzing the dependencies of the memory traces. Even with compiler-aided supports, the inspection-execution technique 
incurs considerable overhead at each instance of the execution and thus increases the overall runtime, instead of reducing 
it to the extent achieved by compile-time optimization approaches. 

Instead of reorganizing access to non-zero data through indirections or leaving its identification to 
runtime, it is desirable to symbolically identify the non-zero computations at compile-time and generate 
code that is aware of the sparsity structure. The state-of-the-art that employs a static approach can be 
divided into two broad categories:
\begin{itemize}
	\item A method could focus only on the sparsity structure of the input, thereby avoiding reading 
		the zero values in the input wherever possible. 
		This approach works only when the output is dense and the memory trace is  
		dominated by the sparsity structure of a single sparse data. Augustine et al.~\cite{piecewise2} and 
		Rodr{\'i}guez et al.~\cite{piecewise1} presented a trace-based technique to generate sparsity structure-specific code for 
		matrix operations resulting in dense data.
	\item A method could focus on the sparsity structure of the output, thereby statically computing the positions of 
		non-zero elements in the output from the sparsity structures of the input. This approach works when the output 
		is sparse or the memory trace is dominated by the sparsity structures of multiple sparse data. This can also 
		handle changes in the sparsity structure of the input, caused by the fill-in elements in mutable cases. 

		Such a method can involve a trace-based technique that simply unwinds a program and parses it based on the input to 
		determine the sparsity structure of the output. Although sounds simple, the complexity of this technique bounds 
		to the computations involved in the matrix operation and size of the output, making it a
		resource-consuming and practically intractable for complex matrix operations and large-sized inputs.
		
		Alternatively, a graph-based technique like {\em Symbolic analysis}~\cite{davis1} uses matrix operation-specific 
		graphs and graph algorithms to deduce the sparsity structure of the output. Cheshmi et 
		al.~\cite{sympiler,8665791,Cheshmi:2018:PIT:3291656.3291739} apply this analysis to collect symbolic information 
		and enables further optimization. The complexity of this technique is bound to the number of non-zero elements of 
		the output, instead of its size, making it significantly less compared to the compile-time trace-based technique. 
		However, the Symbolic analysis is matrix-operation specific, so the customization of the technique to different 
		matrix operations requires manual effort. 
\end{itemize}

We propose a {\em data-flow analysis}-based technique, named {\em Sparsity Structure Specific Compilation} (SpComp), that statically 
deduces the sparsity structure of the output from the sparsity structure of the input. Our method advances the state-of-the-art in 
the following ways. 
\begin{itemize}
	\item In comparison to the run-time approaches, our method does not depend on any run-time information, 
		making it a purely compile-time technique. 
	\item In comparison to the piecewise-regular code generator~\cite{piecewise1,piecewise2}, our method 
		handles matrix operations resulting in sparse output, including mutable cases.
	\item In comparison to the compile-time trace-based technique, the complexity of our method is bound to the number 
		of non-zero elements present in the output which is significantly less than the size of the output, making 
		it a tractable technique.  
	\item In comparison to the Symbolic analysis~\cite{davis1}, our method is generic to any matrix operation, without the 
		need for manual customization. 
\end{itemize}

SpComp takes a program performing matrix operation on dense data and sparsity structures of the input sparse data
to compute the sparsity structure of the output and derive the non-redundant computations. The approach avoids computing 
zero values in the output wherever possible, which automatically implies avoiding reading zero values in the input wherever possible. 
Since it is driven by discovering the sparsity structure of the output, it works for matrix operations producing sparse output 
and altering the sparsity structures of the input. From the derived symbolic information, SpComp generates the sparsity 
structure-specific code, containing piecewise-regular and indirect reference-free loops.

\begin{figure*}[htbp]
\begin{center}
\scalebox{.78}{
\begin{tabular}{cc}
	\begin{tabular}{cc}
		\usebox{\choleskyboxA}
		&
		\scalebox{.9}{
		\begin{tabular}{c}
			\includegraphics[width=4.5cm, height=4.5cm]{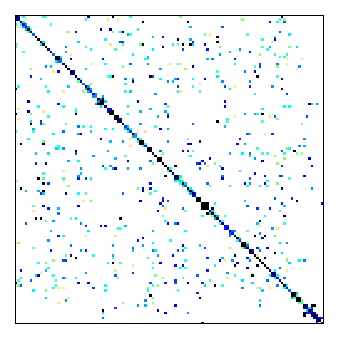}\\
			Size of ${\tt A}$ = ${\tt 494\times 494}$ \\
			NNZ of ${\tt A}$ = ${\tt 1666}$ \\
			Initial sparsity structure of ${\tt A}$ =\\
			{\tt \{(0,0),(0,491),(1,1),(1,491),(2,2),(2,3),}\\
			{\tt (3,2),(3,3),(3,4),(3,491),(4,3),(4,4),$\ldots$\}} \\
		\end{tabular}} \\
		(i) & (ii)
	\end{tabular}
	& 
	\begin{tabular}{c}
		Fill-in elements of ${\tt A}$ =\\
		{\tt \{(38,28),(38,29),(38,30),(38,33),}\\
		{\tt (38,34),(38,36),(38,37),(46,38),}\\
		{\tt (79,78),(82,67),(82,72),(82,73),$\ldots$\}} \\
		(i) \\
		\usebox{\choleskyboxB} \\
		(ii)
	\end{tabular}
	\\
	(a) & (b) \\
\end{tabular}}
\caption{(a) Input to SpComp: (i) Code for Cholesky decomposition operating on symmetric positive definite dense matrix ${\tt A}$,
	(ii) Initial sparsity structure of input matrix ${\tt A}$.
	(b) Output of SpComp: (i) Statically identified Fill-in elements,
	(ii) Snippet of Cholesky decomposition code customized to the sparsity structure of the output matrix ${\tt A}$, ${\tt A_{val}}$ is a one-dimensional 
	array storing the values of non-zero elements of ${\tt A}$.} 
\label{fig:motiv1}
\end{center}
\end{figure*}

\begin{example}{\rm
	Figure~\ref{fig:motiv1}(a) shows the inputs to the SpComp;
	\begin{inparaenum}[(i)]
	\item the code performing forward Cholesky decomposition of symmetric positive definite dense matrix ${\tt A}$ and 
	\item the initial sparsity structure of the input matrix ${\tt 494\_bus}$ selected from the
		\textit{Suitesparse Matrix Collection}~\cite{spCollection}. 
	\end{inparaenum}
	Figure ~\ref{fig:motiv1}(b) illustrates the output of SpComp;
	\begin{inparaenum}[(i)]
	\item fill-in elements of ${\tt 494\_bus}$ along with the initial sparsity structure generate the 
		sparsity structure of the output matrix. 
	\item Cholesky decomposition code, customized to ${\tt 494\_bus}$ sparse matrix. 
	\end{inparaenum}

	Note that, although there exists a read-after-write (true) dependency from statement ${\tt S_3}$ to statement 
	${\tt S_4}$ in the program present in Figure~\ref{fig:motiv1}(ai), 
	a few instances of ${\tt S_4}$ hoist above ${\tt S_3}$ in the execution due to 
	spurious dependencies produced by zero-value computations.  
}\framebox{}
\end{example}

The rest of the paper is organized as follows. Section~\ref{sec:overview} provides an overview of SpComp. 
Section~\ref{sec:technical} describes the first step of SpComp, which identifies the indices involving the 
computations leading to non-zero values through a novel data flow analysis called Essential Indices Analysis.
Section~\ref{sec:codegen} explains the second step which generates the code. Section~\ref{sec:result} presents 
the empirical results. Section~\ref{sec:relwork} describes the related work. Section~\ref{sec:conclusion} concludes 
the paper. 

\section{An Overview of SpComp}
\label{sec:overview}
\begin{figure}[htbp]
\begin{center}
\begin{tabular}{c}
	\includegraphics[width=7.5cm, height=1.8cm]{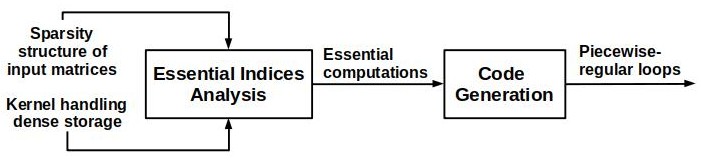}
	\end{tabular}
	\caption{Block diagram of SpComp.}
	\label{fig:block}
\end{center}
\end{figure}

As depicted in Figure~\ref{fig:block}, SpComp has two modules performing (a) \textit{essential indices analysis} 
and (b) \textit{code generation}. The essential indices analysis module constructs an \textit{Access Dependence Graph} (ADG) 
from the program and performs a data flow analysis, named \textit{Essential Indices Analysis}. The analysis effectively 
identifies the \textit{essential data indices} of the output matrix and \textit{essential iteration indices} of the 
iteration space. Essential data indices identify the indices of non-zero elements which construct the underlying sparse 
data storage beforehand, without requiring any modification during run-time. Essential iteration indices identify the 
iteration points in the iteration space that must execute to compute the values of the non-zero elements and facilitate 
the generation of piecewise-regular loops.

\begin{figure}[htbp]
	\begin{center}
		\scalebox{.9}{
		\begin{tabular}{ccccc}
			Fill-in & \multicolumn{4}{c}{Essential iteration indices of} \\
			elements of {\tt A} & ${\tt S_1}$ & ${\tt S_2}$ & ${\tt S_3}$ & ${\tt S_4}$ \\
			\hline 
			${\tt (38,27)}$ & ${\tt (9,8,7)}$ & ${\tt (3,2)}$ & ${\tt (3,2)}$ & ${\tt (0)}$ \\
			${\tt (38,28)}$ & ${\tt (38,28,27)}$ & ${\tt (4,3)}$ & ${\tt (4,3)}$ & ${\tt (1)}$ \\
			${\tt (38,29)}$ & ${\tt (38,29,28)}$ & ${\tt (5,4)}$ & ${\tt (5,4)}$ & ${\tt (2)}$ \\
			${\tt (38,32)}$ & ${\tt (38,30,29)}$ & ${\tt (6,5)}$ & ${\tt (8,7)}$ & ${\tt (3)}$ \\
			${\tt (38,33)}$ & ${\tt (38,33,32)}$ & ${\tt (8,7)}$ & ${\tt (9,6)}$ & ${\tt (4)}$ \\
			$\hdots$ & $\hdots$ & $\hdots$ & $\hdots$ & $\hdots$ 
		\end{tabular}}
		\caption{Essential data indices of sparse matrix {\tt A} and essential iteration indices of 
		statements ${\tt S_1}$, ${\tt S_2}$, ${\tt S_3}$, and ${\tt S_4}$ generated by the essential 
		analysis module for the example shown in Figure~\ref{fig:motiv1}. Fill-in elements with initial 
		non-zero elements of {\tt A} construct the set of essential data indices. }
	\label{fig:motiv2}
	\end{center}
\end{figure}

\begin{example}{\rm
For our motivating example in Figure~\ref{fig:motiv1}, the essential indices analysis generates the essential data 
indices and essential iteration indices for statements ${\tt S_1}$, ${\tt S_2}$, ${\tt S_3}$, and ${\tt S_4}$ as shown 
in Figure~\ref{fig:motiv2}. Note that this analysis is an abstract interpretation of the program with
abstract values, we do not compute the actual expressions with concrete values. The input sparse matrix {\tt 494\_bus} 
is represented by the array {\tt A} in the code, which is both the input and the output matrix. 

The essential data indices of the output matrix comprise the indices of non-zero elements of the input matrix and the 
\textit{fill-in} elements denoting the indices whose values are zero in the input but become non-zero in the output.
\footnote{The converse (i.e. a non-zero value of an index in the input becoming zero at the same index in the output) 
is generally not considered explicitly in such computations and are performed by default.}
Fill-in element ${\tt (38,27)}$ identifies ${\tt A[38][27]}$ whose value changes from zero to non-zero during the 
execution of the program. The essential iteration index ${\tt (0)}$ of statement ${\tt S_4}$ identifies 
${\tt A[0][0]=sqrt(A[0][0])}$ as a statement instance that should be executed to preserve the semantics of the program. 
}\framebox{}
\end{example}

At first, the code generation module finds the timestamps of the essential iteration indices and lexicographically orders them to 
construct the execution trace, without any support for the out-of-order execution. Then it simply finds the pieces of execution trace 
that can be folded back into regular loops. The module also constructs the memory access trace caused by the execution order and mines 
the access patterns for generating the subscript functions of the regular loops. Note that, the generated code 
keeps the {\em if} conditions to avoid division by zero during execution. 
\begin{example}{\rm
The execution trace $\langle {\tt S_4,0}\rangle$ $\to$ $\langle{\tt S_4,1}\rangle$ $\to$ $\langle{\tt S_4,2}\rangle$ 
$\to$ $\langle{\tt S_2,3,2}\rangle$ $\to$ $\langle{\tt S_3,3,2}\rangle$ $\to$ $\langle{\tt S_4,3}\rangle$ $\to$ $\langle{\tt S_2,4,3}\rangle$ 
$\to$ $\langle{\tt S_3,4,3}\rangle$ $\to$ $\ldots$ is generated by the lexicographic order of the timestamps where 
$\langle{\tt S_k,i,j}\rangle$ represents iteration index ${\tt (i,j)}$ of statement~${\tt S_k}$. It is evident that the 
piece of execution trace $\langle {\tt S_4,0}\rangle$ $\to$ $\langle{\tt S_4,1}\rangle$ $\to$ $\langle{\tt S_4,2}\rangle$ 
can be folded back in a loop. The corresponding memory access trace ${\tt A[0][0]}$ $\to$ ${\tt A[1][1]}$ $\to$ ${\tt A[2][2]}$ 
creates a one-dimensional subscript function ${\tt valA[2\times i+0]|0\le i \le 2}$. 
${\tt valA}$ represents the one-dimensional array storing the non-zero values of ${\tt A}$ and ${\tt 2\times\!i+0|0\le i\le 2}$ 
represents the positions of ${\tt A[0][0]}$, ${\tt A[1][1]}$, and ${\tt A[2][2]}$ in the sparse data storage. The generated 
code snippet is presented in Figure~\ref{fig:motiv1}(bii). 
}\framebox{}
\end{example}

\section{Essential Indices Analysis}
\label{sec:technical}
In sparse matrix operations, we assume that the default values of matrix elements are zero. The efficiency of a sparse matrix operation 
lies in avoiding the computations leading to zero or default values. We call such computations as default computations. Computations 
leading to non-zero values are non-default computations.

We refer to the data indices of all input and output matrices holding non-zero values as {\em essential data indices}. 
As mentioned before, we have devised a data flow analysis technique called {\em Essential Indices analysis} 
that statically computes the essential data indices of output matrices from the essential data indices of input matrices. 
We identify all the iteration indices of the loop computing non-default computations as {\em essential iteration indices}. 
Here we describe the analysis by defining {\em Access Dependence Graph} as the data flow analysis graph in Subsection~\ref{sec:adg}, 
the domain of data flow values in Subsection~\ref{sec:dfv}, and the transfer functions with the data 
flow equations in Subsection~\ref{sec:tf}. Finally, we prove the correctness of the analysis in Subsection~\ref{sec:correctness}.

\subsection{Access Dependence Graph}
\label{sec:adg}
Conventionally, data flow analysis uses the {\em Control Flow Graph} (CFG) of a program to compute data flow information at 
each program point. However, CFG is not suitable for our analysis because the set of information computed over CFG of a loop is an
over-approximation of the union of information generated in all iterations. Thus, information gets conflated across all
iterations, and no distinction exists between the fact that information generated in ${\tt i}$-th iteration cannot be used in iterations 
${\tt j}$ if ${\tt j\leq i}$.

SpComp accepts \textit{static control parts} (SCoP)~\cite{10.1007/978-3-642-11970-5_16,10.1007/978-3-540-24644-2_14} 
of a program which is a sequence of perfectly and imperfectly nested loops where loop bounds and subscript 
functions are affine functions of surrounding loop iterators and parameters. For essential indices analysis, 
we model SCoP in the form of an {\em Access Dependence Graph} (ADG). ADG captures 
\begin{inparaenum}[(a)]
\item data dependence, i.e., accesses of the same memory locations, and
\item data flow, i.e., the flow of a value from one location to a different location.
\end{inparaenum}
They are represented by recording flow, anti and output data dependencies, and the temporal order of read and write 
operations over distinct locations. This modeling of dependence is different from the modeling of dependence in a 
{\em Data Dependence Graph} (DDG)~\cite{10.5555/502981,10.5555/535430} that models data dependencies among loop 
statements which are at a coarser level of granularity.
	 
ADG captures dependencies among access operations which are at a finer level of granularity compared 
to loop statements. Access operations on concrete memory locations, i.e., the locations created at run-time, 
are abstracted by access operations on abstract memory locations that conflate concrete memory locations 
accessed by a particular array access expression. For example, the write access operations on concrete memory locations at statement 
${\tt S_1}$ in Figure~\ref{fig:motiv1}(a) are accessed by the access expression ${\{\tt A[i][j]|0\le i<n, 0\le j<i\}}$.
	
ADG handles the affine subscript function of the form ${\tt \sum\limits_{k=1}^n{a_k\times i_k + c}}$, assuming ${\tt a_k}$ and ${\tt c}$ be 
constants and ${\tt i_k}$ be an iteration index. A set of concrete memories read by an affine array expression 
${\{\tt A[f(i_1,\ldots,i_n)]}$ ${\tt |lb_l\le i_l<ub_l, 1\le l<n}\}$ at statement ${\tt S_k}$ is denoted by an access operation 
${\tt r^k_{A[f(i_1,\ldots,i_n)]}}$, where ${\tt lb_l}$ and ${\tt ub_l}$ denote lower and upper bounds of regular or irregular loops. 
Similarly, a set of concrete memories written by the same array expression at statement ${\tt S_l}$ is denoted by an access operation 
${\tt w^l_{A[f(i_1,\ldots,i_n)]}}$. Note that, the bounds on the iteration indices ${\tt i_1,\ldots,i_n}$ become implicit 
to the access operation. For code generation, we concretize an abstract location ${\tt A[i][j]}$ into concrete memory 
locations ${\tt A[1][1]}$, ${\tt A[1][2]}$ etc. using the result of essential indices analysis as explained in Section~\ref{sec:tf}.

ADG captures the temporal ordering of access operations using edges annotated with a dependence direction that models 
the types of dependencies. Dependence direction ${\tt <}$, ${\tt \leq}$ and ${\tt <}$ model flow, anti and output 
data dependencies, whereas dependence direction ${\tt =}$ captures the data flow between distinct memory locations. 
Note that dependence direction ${\tt >}$ is not valid as the source of a dependency can not be executed after the target. 

\begin{figure}[htbp]
	\begin{center}
	\scalebox{.9}{
	\begin{tabular}{cc}
		\begin{tabular}{c}
		\usebox{\mvbox}
		\end{tabular}
	 	&
		\scalebox{.75}{
	 	\begin{tabular}{c}
	 	\includegraphics[width=3.6cm, height=4cm]{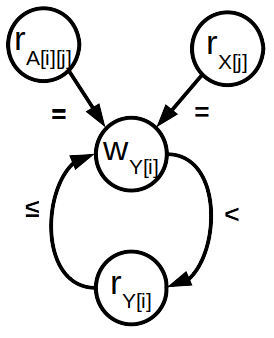}
		\end{tabular}} \\
	 	(a) & (b)
	\end{tabular}}
	\end{center}
	\caption{(a) An SCoP of Matrix-Vector Multiplication program operating on dense data and (b) The corresponding ADG.}
	\label{fig:adg1}
\end{figure}

\begin{example}{\rm
Consider statement \text{${\tt Y[i]\!=\!Y[i]\!+\!A[i][j]\!\times\! X[j]}$}
in Figure~\ref{fig:adg1}(a). It is evident that there is an anti dependency from the read access of ${\tt \{Y[i]|0\le i<n\}}$, 
denoted as ${\tt r_{Y[i]}}$, to write access of ${\tt\{Y[i]|0\le i<n\}}$, denoted as ${\tt w_{Y[i]}}$. 
The anti-dependency is represented by an edge \text{${\tt r_{Y[i]}} \xrightarrow{\leq} {\tt w_{Y[i]}}$} where 
the direction of the edge indicates the ordering and the edge label $\leq$ indicates that it is an anti-dependency.
	 
Similarly, there is a flow dependency from ${\tt w_{Y[i]}}$ to ${\tt r_{Y[i]}}$. This is represented by an 
edge \text{${\tt w_{Y[i]}} \xrightarrow{<} {\tt r_{Y[i]}}$} where direction of the edge indicates the ordering 
and the edge label ${\tt <}$ denotes the flow dependency.
	 
The data flow from ${\tt r_{A[i][j]}}$ to ${\tt w_{Y[i]}}$ and ${\tt r_{X[j]}}$ to ${\tt w_{Y[i]}}$ are denoted 
by the edges \text{${\tt r_{A[i][j]}} \xrightarrow{=} {\tt w_{Y[i]}}$} and \text{${\tt r_{X[j]}} \xrightarrow{=} {\tt w_{Y[i]}}$} 
respectively where the directions indicate the ordering and the edge labels $=$ indicate that these are data flow.
}\framebox{}
\end{example}

Each vertex ${\tt v}$ in the ADG ${\tt G=(V, E)}$ is called an access node, and the entry and exit points of each 
access node are access points. Distinctions between entry and exit access points are required for formulating 
transfer functions and data flow equations in Section~\ref{sec:tf}.

Formally, ADG captures the temporal and spatial properties of data flow. Considering every statement instance 
as atomic in terms of time, the edges in an ADG have associated temporal and spatial properties, as explained below. 
	 
Let ${\tt r_P}$ and ${\tt w_Q}$ respectively denote read and write access nodes. 
\begin{itemize}
	\item For edge \text{${\tt r_P}\xrightarrow{=}{\tt w_Q}$} where ${\tt mem(r_P)\cap mem(w_Q)=\emptyset}$, edge label ${\tt =}$ 
	implies that ${\tt w_Q}$ executes in the same statement instance as ${\tt r_P}$ and captures data flow.
	\item For edge \text{${\tt r_P}\xrightarrow{\leq}{\tt w_Q}$} where ${\tt mem(r_P)\cap mem(w_Q)\neq\emptyset}$, 
	edge label ${\tt \leq }$ implies that ${\tt w_Q}$ executes either in the same statement instance as ${\tt r_P}$ or 
	in a statement instance executed later and captures anti dependencies. 
	\item For edge \text{${\tt w_P}\xrightarrow{<}{\tt r_Q}$} where ${\tt mem(w_P)\cap mem(r_Q)\neq\emptyset}$, edge label ${\tt <}$ 
	implies that ${\tt r_Q}$ executes in a statement instance executed after 
	the execution of ${\tt w_P}$ and captures flow dependencies. 
	\item For edge \text{${\tt w_P}\xrightarrow{<}{\tt w_Q}$} where ${\tt mem(w_P)\cap mem(w_Q)\neq\emptyset}$, edge label ${\tt <}$ 
	implies that ${\tt w_Q}$ executes in a statement instance executed after 
	the execution of ${\tt w_P}$ and captures output dependencies. 
	\end{itemize}

\subsection{Domain of Data Flow Values}
\label{sec:dfv}
Let the set of data indices of an ${\tt n}$-dimensional matrix ${\tt A}$ of size ${\tt m_1\times\ldots\times m_n}$ be represented
as $\mathcal{D}{\tt_A^n}$ where ${\tt A}$ has data size ${\tt m_k}$ at dimension ${\tt k}$. Here 
$\mathcal{D}{\tt_A^n=\{}\vec{\tt d}{\tt |(0,\ldots,0)\le}\ \vec{\tt d}\ {\tt \le (m_1,\ldots,m_n)\}}$ 
where vector $\vec{\tt d}$ represents a data index of matrix ${\tt A}$.
If ${\tt A}$ is sparse in nature then $\vec{\tt d}$ is an essential data index if ${\tt A[}\vec{\tt d}{\tt ]\ne 0}$. 
The set of essential data indices of sparse matrix ${\tt A}$ is represented as $\mathbb{D}{\tt _A^n}$ such that 
$\mathbb{D}{\tt _A^n}\subseteq\mathcal{D}{\tt_A^n}$. For example, $\mathcal{D}{\tt _A^2}$ of a ${\tt two}$-dimensional 
matrix of size ${\tt 3\times 3}$ is ${\tt \{}$${\tt(0,0)}$, ${\tt(0,1)}$, ${\tt(0,2)}$, ${\tt(1,0)}$, ${\tt(1,1)}$, 
${\tt(1,2)}$, ${\tt(2,0)}$, ${\tt(2,1)}$, ${\tt(2,2)\}}$, the set of all data indices of ${\tt A}$. If ${\tt A}$ is 
sparse with non-zero elements at ${\tt A[0][0]}$, ${\tt A[1][1]}$, and ${\tt A[2][2]}$ then $\mathbb{D}{\tt _A^2}$ 
= ${\tt \{(0,0), (1,1), (2,2)\}}$. Thus $\mathbb{D}{\tt _A^2}\subseteq\mathcal{D}{\tt _A^2}$.
	 
In this analysis, we consider the union of all data indices of all input and output matrices as \textit{data space} ${\mathcal D}$.
The union of all essential data indices of all input and output matrices is considered as 
the {\em domain of data flow values} ${\mathbb D}$ such that $\mathbb{D}\subseteq{\mathcal D}$.
For our analysis each essential data index of ${\mathbb D}$ is annotated with the name of the origin matrix. For
an ${\tt n}$-dimensional matrix ${\tt A}$ each essential data index thus represents an ${\tt n+1}$ dimensional vector 
$\vec{\tt d}$ where the ${\tt 0}$-th position holds the name of the matrix, ${\tt A}$. For example, in the matrix-vector 
multiplication of Figure~\ref{fig:adg1}(a), let ${\mathbb D}{\tt _A^2=\{(0,0),(1,1),(2,2)\}}$, ${\mathbb D}{\tt _X^1=\{(1),(2)\}}$ 
and ${\mathbb D}{\tt _Y^1=\{(1),(2)\}}$. Thus, the domain of data flow values ${\mathbb D}$ is ${\tt=\{(A,0,0),}$ ${\tt(A,1,1),}$ 
${\tt(A,2,2),}$ ${\tt(X,1),}$ ${\tt(X,2),}$ ${\tt(Y,1),}$ ${\tt(Y,2)\}}$.
 
In the rest of the paper data index $\vec{\tt d}$ is denoted as ${\tt d}$ for convenience. The value at data 
index ${\tt d}$ is abstracted as either zero ({\tt Z}) or nonzero ({\tt NZ}) based on the concrete value at ${\tt d}$. Note that 
the domain of concrete values, ${\tt cval}$, at each data index ${\tt d}$ is a power set of $\mathbb{R}$, which is the set of real numbers. 
Our approach abstracts the concrete value at each ${\tt d}$ by ${\tt val(d)}$ defined as follows. 
	\begin{eqnarray}
		{\tt val(d)} = 
	\begin{cases}
		{\tt Z} &\quad\text{if }{\tt cval(d) = \{0\}} \\
		{\tt NZ} & \quad\text{otherwise}
	\end{cases}
	\end{eqnarray}
	 
The domain of values at each data index ${\tt d}$ forms a component lattice $\hat{\mathbb{L}}=\langle\{{\tt Z, NZ}\},\sqsubseteq\rangle$, 
where ${\tt NZ}\sqsubseteq{\tt Z}$ and $\sqsubseteq$ represents the partial order. 
	 
A data index ${\tt d}$ is called \textit{essential} if ${\tt val(d)=NZ}$. As the data flow value at any access point holds the set 
of essential data indices $\mathbb{D}^\prime$ such that $\mathbb{D}^\prime\subseteq 2^\mathcal{D}$, 
the data flow lattice is thus represented as $\langle 2^\mathcal{D},\supseteq\rangle$ where partial order is a superset relation.
	
\subsection{Transfer Functions}
\label{sec:tf}
This section formulates the transfer functions used to compute the data flow information of all data flow variables and presents 
the algorithm performing Essential Indices Analysis. For an ADG, ${\tt G=(V,E)}$, the data flow variable ${\tt Gen_n}$
captures the data flow information generated by each access node ${\tt n\in V}$, and the data flow variable ${\tt Out_n}$ 
captures the data flow information generated at the exit of each node ${\tt n\in V}$.

Let $\mathbb{D}{\tt ^0}$ be the initial set of essential data indices that identifies the indices of non-zero elements of input matrices. 
In Essential Indices Analysis, ${\tt Out_n}$ is defined as follows.
	\begin{eqnarray} 
	\label{eq:out}
		{\tt Out_n} = 
	\begin{cases}
		\quad\mathbb{D}{^0} & \quad\text{if }{\tt Pred_n}=\emptyset \\
		\bigcup\limits_{{\tt p\in Pred_n}}({\tt Out_p\cup Gen_n})& \quad \text{if }{\tt n}\text{ is write} \\
		\bigcup\limits_{{\tt p\in Pred_n}}{\tt Out_p} & \quad \text{otherwise} 
	\end{cases}
	\end{eqnarray}
	${\tt Pred_n}$ denotes the set of predecessors of each node ${\tt n}$ in the access dependence graph.
	 
Equation~\ref{eq:out} initializes ${\tt Out_n}$ to $\mathbb{D}{\tt ^0}$ for access node ${\tt n}$ that does not have any predecessor. 
Read access nodes do not generate any essential data indices; they only combine the information of their predecessors. 
A write access node typically computes the arithmetic expression associated with it and generates a set of essential data 
indices as ${\tt Gen_n}$. Finally, ${\tt Gen_n}$ is combined with the out information of its predecessors to compute 
${\tt Out_n}$ of write node ${\tt n}$. 
	 
Here we consider the statements associated with write access nodes and admissible in our analysis. They are of the form 
${\tt e_1:d=d^\prime}$, ${\tt e_2:d=op(d^\prime)}$ and ${\tt e_3:d=op(d^\prime,d^{\prime\prime})}$ where ${\tt e_1}$ is 
copy assignment, ${\tt e_2}$ uses unary operation and ${\tt e_3}$ uses binary operation. Here ${\tt d}$, ${\tt d^\prime}$, 
and ${\tt d^{\prime\prime}}$ are data indices. 
	 
Instead of concrete values, the operations execute on abstract values $\{{\tt Z,NZ}\}$. Below we present the evaluation of 
all valid expressions on abstract values. Note that the unary operations negation, square root, etc. return the same abstract 
value as input.\footnote{Unary operations such as floor, ceiling, round off, truncation, saturation, shift
etc. that may change the values are generally not used in sparse matrix operations.}
Thus evaluation effect of expressions {\tt $e_1$} and {\tt $e_2$} are combined into the following.
	\begin{equation}
		\label{eq:op1}
		{\tt val(d) = val(d^\prime)}
	\end{equation}
	 
We consider the binary operations addition, subtraction, multiplication, division, and modulus. Thus the 
evaluation of expression ${\tt e_3}$ is defined as follows for the aforementioned arithmetic operations.
	\begin{itemize}
		\item If ${\tt op}$ is addition or subtraction then 
	\begin{align}
		\label{eq:op2}
		{\tt val(d)} = 
	\begin{cases}
		{\tt NZ}& \text{if } {\tt val(d^\prime)}={\tt NZ}\vee {\tt val(d^{\prime\prime})}={\tt NZ} \\
	{\tt Z} & \text{otherwise}
	\end{cases}
	\end{align}
	\item If ${\tt op}$ is multiplication then 
	\begin{align}
		\label{eq:op3}
		{\tt val(d)} = 
	\begin{cases}
		{\tt NZ}&\text{if } {\tt val(d^\prime)}={\tt NZ}\wedge {\tt val(d^{\prime\prime})}={\tt NZ} \\
	{\tt Z}& \text{otherwise}
	\end{cases}
	\end{align}
	\item If ${\tt op}$ is division or modulus then 
	\begin{align}
		\label{eq:op4}
		{\tt val(d)} = {\tt val(d^\prime)}\quad\text{if } {\tt val(d^{\prime\prime})\neq Z}
	\end{align}
	\end{itemize}
	 
Note that, the addition and subtraction operations may result in zero values due to numeric cancellation while 
operating in the concrete domain. This means {\tt cval(d)} can be zero when ${\tt cval(d^\prime)}$ and ${\tt cval(d^{\prime\prime})}$
are non-zeroes. In this case, our abstraction over-approximates ${\tt cval(d)}$ as ${\tt NZ}$, which is a safe approximation.

Division or modulus by zero is an undefined operation in the concrete domain and is protected by the condition on the denominator value 
${\tt cval(d^{\prime\prime})\neq 0}$. We protect the same in the abstract domain by the condition ${\tt val(d^{\prime\prime})\neq Z}$, 
as presented in Equation~\ref{eq:op4}. Although handled, the conditions are still part of the generated code to preserve the semantic 
correctness of the program. For example, the {\em if} conditions in the sparsity structure-specific Cholesky decomposition code in 
Figure~\ref{fig:motiv1}(bii) preserve the semantic correctness of the program by prohibiting the divisions by zero values in the 
concrete domain. In this paper, we limit our analysis to simple arithmetic operations. However, similar abstractions could be defined 
for all operations by ensuring that no possible non-zero result is abstracted as zero.
	 
Now that we have defined the abstract value computations for different arithmetic expressions, 
below we define ${\tt Gen_n}$, which generates the set of essential data indices for write access node ${\tt n}$.
	\begin{eqnarray}
	\label{eq:gen}
		{\tt Gen_n} = 
	\begin{cases}
		\big\{{\tt d}\quad | & {\tt(d=d^\prime\vee d=op(d^\prime))}\\
		&\wedge{\tt(d}^\prime\in{\tt Out_p}, {\tt p\in Pred_n)}\\
		&\wedge {\tt(val(d)=NZ})\big\}\\
		\big\{{\tt d}\quad | & ({\tt d}={\tt op(d}^\prime,{\tt d}^{\prime\prime}))\\
		&\wedge({\tt d}^\prime\in{\tt Out_p}, {\tt p}\in {\tt Pred_n}\\
		&\vee{\tt d}^{\prime\prime}\in{\tt Out_q}, {\tt q}\in{\tt Pred_n})\\
		&\wedge {\tt(val(d)}={\tt NZ})\big\} \\
		\emptyset & \text{otherwise}
	\end{cases}
	\end{eqnarray}
In the case of a binary operation, predecessor node ${\tt p}$ may or may not be equal to predecessor node ${\tt q}$.
	
The objective is to compute the least fixed point of Equation~\ref{eq:out}. Thus the analysis must begin 
with the initial set of essential data indices $\mathbb{D}{\tt ^0}$. The data flow variables ${\tt Out_n}$ 
are set to initial values and the analysis iteratively computes the equations until the fixed point is reached. 
If we initialize with something else the result would be different and it would not be the least fixed point.
Once the solution is achieved the analysis converges. 
	 
Essential indices analysis operates on finite lattices and is monotonic as it only adds the generated information 
in each iteration. Thus, the analysis is bound to converge on the fixed point solution.
The existence of a fixed point solution is guaranteed by the finiteness of lattice and monotonicity of flow functions.
\begin{figure}[htbp]
\begin{center}
\scalebox{.85}{
\begin{tabular}{|c|c|c|c|c|}
\multicolumn{5}{l}{\text{Let, }$\mathbb{D}^{\tt 0}{\tt =\{(A,0,0),(A,0,2),(A,1,1),(A,2,1),(A,3,1),(A,3,3),(X,1),(X,3)\},}$}\\
\multicolumn{5}{l}{$\mathbb{D}^{\prime}{\tt =\{(Y,1),(Y,2),(Y,3)\}}\text{ and }\mathbb{D}^{\tt 1=}\mathbb{D}^{\tt 0}\cup\mathbb{D}^\prime$} \\
\multicolumn{5}{c}{}\\
\hline
\multicolumn{5}{|c|}{Essential indices analysis}\\
\hline
Data flow variable & Initialization & Iteration ${\tt 1}$ & Iteration ${\tt 2}$ & Iteration ${\tt 3}$ \\
\hline
\hline
${\tt Gen_{r_{A[i][j]}}}$ & ${\tt -}$ & $\emptyset$ & $\emptyset$ & $\emptyset$\\
${\tt Out_{r_{A[i][j]}}}$ & $\mathbb{D}^{\tt 0}$ & $\mathbb{D}^{\tt 0}$ & $\mathbb{D}^{\tt 0}$ & $\mathbb{D}^{\tt 0}$ \\
\hline
${\tt Gen_{r_{X[j]}}}$ & ${\tt -}$ & $\emptyset$ & $\emptyset$ & $\emptyset$\\
${\tt Out_{r_{X[j]}}}$ & $\mathbb{D}^{\tt 0}$ & $\mathbb{D}^{\tt 0}$ & $\mathbb{D}^{\tt 0}$ & $\mathbb{D}^{\tt 0}$ \\
\hline
${\tt Gen_{r_{Y[i]}}}$ & ${\tt -}$ & $\emptyset$ & $\emptyset$ & $\emptyset$\\
${\tt Out_{r_{Y[i]}}}$ & $\mathbb{D}^{\tt 0}$ & $\mathbb{D}^{\tt 0}$ & $\mathbb{D}^{\tt 1}$ & $\mathbb{D}^{\tt 1}$ \\
\hline
${\tt Gen_{w_{Y[i]}}}$ & ${\tt -}$ & $\mathbb{D}^\prime$ & $\emptyset$ & $\emptyset$\\
${\tt Out_{w_{Y[i]}}}$ & $\emptyset$ & $\mathbb{D}^{\tt 1}$ & $\mathbb{D}^{\tt 1}$ & $\mathbb{D}^{\tt 1}$ \\
\hline
\end{tabular}}
\end{center}
\caption{Essential indices analysis of sparse matrix-sparse vector multiplication ${\tt Y=A\times X}$ where ${\tt A}$, ${\tt X}$ 
and ${\tt Y}$ are sparse. Set of essential data indices of ${\tt A}$ and ${\tt X}$ are ${\tt \{(0,0),(0,2),(1,1),(2,1),(3,1),(3,3)\}}$ 
and ${\tt \{(1),(3)\}}$ respectively.}
\label{fig:analysis}
\end{figure}

\begin{example}{\rm
Figure~\ref{fig:analysis} demonstrates essential indices analysis for the sparse matrix-sparse vector multiplication operation. 
It takes the ADG from Figure~\ref{fig:adg1}(b) and performs the data flow analysis based on the sparsity structures of input 
matrix ${\tt A}$ and input vector ${\tt X}$ as depicted in Figure~\ref{fig:motiv}(a). 
Let, $\mathbb{D}^{\tt 0}$ be ${\tt\{(A,0,0)}$, ${\tt(A,0,2)}$, ${\tt(A,1,1)}$, ${\tt(A,2,1)}$, ${\tt(A,3,1)}$, ${\tt(A,3,3)}$, 
${\tt(X,1)}$, ${\tt(X,3)\}}$. The gen and out information of access nodes ${\tt r_{A[i][j]}}$,
${\tt r_{X[j]}}$, ${\tt r_{Y[i]}}$ and ${\tt w_{Y[i]}}$ are presented in tabular form for convenience. Out information of all read 
nodes are initialized to $\mathbb{D}^{\tt 0}$, whereas the out information of write node is initialized to $\emptyset$. 
	 
In iteration ${\tt 1}$, ${\tt Gen_{Y[i]}}$ is $\mathbb{D}^\prime$ where $\mathbb{D}^\prime{\tt =\{(Y,1),(Y,2),(Y,3)\}}$. Thus 
${\tt Out_{Y[i]}}$ becomes $\mathbb{D}^{\tt 1}$ such that ${\tt \mathbb{D}^1=\mathbb{D}^0\cup\mathbb{D}^\prime}$. In iteration 
${\tt 2}$ the out information of ${\tt w_{Y[i]}}$ propagates along with edge ${\tt w_{Y[i]}\to r_{Y[i]}}$ and sets 
${\tt Out_{r_{Y[i]}}}$ to $\mathbb{D}^{\tt 1}$. Finally at iteration ${\tt 3}$ the analysis reaches the fixed point 
solution and converges.
}\framebox{}
\end{example}
	 
Data flow variable ${\tt AGen_n}$ is introduced to accumulate ${\tt Gen_n}$ at each write node ${\tt n}$ required 
by a post-analysis step computing set of essential iteration indices from the set of essential data indices. 
	\begin{equation}
	\label{eq:agen}
		{\tt AGen_n} = {\tt AGen_n}\cup{\tt Gen_n} 
	\end{equation}
In the current example ${\tt AGen_{w_{Y[i]}}}$ is initialized to $\emptyset$. It accumulates 
${\tt Gen_{w_{Y[i]}}}$ generated at each iteration, finally resulting ${\tt AGen_{w_{Y[i]}}}$ as ${\tt \{(Y,1),(Y,2),(Y,3)\}}$.

${\tt AFill}$ denotes the fill-in elements of the output sparse matrix and is computed as follows. 
	\begin{equation}
	\label{eq:afill}
		{\tt AFill = \bigcup_{\forall n\in V}AGen_n \setminus \mathbb{D}^0} 
	\end{equation}

The initial essential data indices ${\tt \mathbb{D}^0}$ and the fill-in elements ${\tt AFill}$ together compute the final 
essential data indices ${\tt \mathbb{D}^f}$ that captures the sparsity structure of the output matrix.
	\begin{equation}
		\label{eq:sp}
		{\tt \mathbb{D}^f = AFill \cup \mathbb{D}^0}
	\end{equation}
In the current example, ${\tt AFill}$ is computed as ${\tt \{(Y,1),(Y,2),(Y,3)\}}$. As ${\tt \mathbb{D}^0}$ does not contain any initial 
essential data index of ${\tt Y}$, ${\tt\mathbb{D}^f}$ becomes same as ${\tt AFill}$.

The set of essential iteration indices is computed from the set of essential data indices. 
Let ${\mathcal I}$ of size ${\tt l_1\times\ldots\times l_p}$ be the iteration space
of dimension ${\tt p}$ of a loop having depth ${\tt p}$ where the loop at depth ${\tt k}$ has iteration size ${\tt l_k}$. Thus, 
$\mathcal{I}=\{\vec{\tt i_k}{\tt | 1\le k\le l_1\times\ldots\times l_p}\}$ where vector $\vec{\tt i_k}$
is an iteration index. For convenience, here on we identify $\vec{\tt i_k}$ as ${\tt i}$. The iteration index at which 
non-default computations are performed is called the essential iteration index. 
The set of all essential iteration indices is denoted as $\mathbb{I}$ such that $\mathbb{I}\subseteq\mathcal{I}$.
	 
For each essential data index ${\tt d\in AGen_n}$ there exists a set of essential iteration indices 
${\mathbb{I}^\prime}$ at which the corresponding non-default computations resulting ${\tt d}$ occur. 
We introduce ${\tt iter}:\mathbb{D}\to{\tt 2^\mathbb{I}}$ such that ${\tt iter(d)}$ results $\mathbb{I}^\prime$ 
where ${\tt d\in\mathbb{D}}$ and $\mathbb{I}^\prime\in{\tt 2^\mathbb{I}}$. Data flow variable ${\tt AInd_n}$ 
is introduced to capture the set of essential iteration indices corresponding to the data indices ${\tt d\in AGen_n}$.
Thus, 
	\begin{eqnarray}
	\label{eq:ind}
		{\tt AInd_n} = \bigcup_{\forall {\tt d}\in{\tt AGen_n}}{\tt iter(d)}
	\end{eqnarray}
In the current example ${\tt AInd_{w_{Y[i]}}}$ is computed as ${\tt \{(1,1),}$ ${(2,1),}$ ${(3,1),}$ ${(3,3)\}}$.

Finally, the set of all essential iteration indices ${\tt \mathbb{I}}$ is computed as
	\begin{eqnarray}
	\label{eq:ess_ind}
		{\tt \mathbb{I} = \bigcup_{\forall n\in V} AInd_n}
	\end{eqnarray}
	
	\begin{figure}[t]
	\begin{tabular}{c}
	\hline
	\noindent
	\scalebox{.9}{
	\begin{minipage}[t]{.9\linewidth}
	\begin{algorithm}[H]
	\caption{Algorithm for Essential Indices Analysis.}
	\label{algo:analysis}
	\SetAlgoLined
	\LinesNumbered
	\SetKwInOut{Input}{Input}
	\SetKwInOut{Output}{Output}
	\SetKwRepeat{Do}{do}{while}
	\Input{
		\begin{inparaenum}[(a)]
		\item ADG, ${\tt G(V,E)}$, of a matrix operation containing ${\tt V}$ access nodes and ${\tt E}$ edges,
		\item Intial set of essential data indices ${\tt D^0}$ of input matrices.
		\end{inparaenum}
		}
	\Output{
		\begin{inparaenum}[(a)]
		\item Set of essential data indices ${\tt \mathbb{D}^f}$ of output matrix,
		\item Set of essential iteration indices ${\tt \mathbb{I}}$.
		\end{inparaenum}
		}

		${\tt \forall n\in V}$, initialize ${\tt Out_n}$ to ${\tt D^0}$ where ${\tt Pred_n = \emptyset}$\\
		${\tt WorkList\leftarrow\bigcup\limits_{\forall n\in V}{n}}$ \\
		\Do{${\tt WorkList = \emptyset}$}{ 
		Pick and remove node ${\tt n}$ from ${\tt WorkList}$\\
		${\tt OldOut_n \leftarrow Out_n}$ \\
		Compute ${\tt Gen_n}$ and ${\tt AGen_n}$ using Equation~\ref{eq:gen} and Equation~\ref{eq:agen} respectively.\\
		Compute ${\tt Out_n}$ using Equation~\ref{eq:out}. \\
		\If{${\tt OldOut_n \neq Out_n}$}{${\tt WorkList \leftarrow WorkList \cup n}$}
		}
		Compute ${\tt \mathbb{D}^f}$ and ${\tt \mathbb{I}}$ using Equation~\ref{eq:sp} and Equation~\ref{eq:ess_ind} respectively\\
		\Return
	\end{algorithm}
	\end{minipage}} \\
	\hline
	\end{tabular}
	\end{figure}

Algorithm~\ref{algo:analysis} presents the algorithm for essential indices analysis. Line number ${\tt 1}$ initializes ${\tt Out_n}$. 
Line number ${\tt 4}$ sets the work list, ${\tt WorkList}$, to the nodes of the ADG. Lines ${\tt 3-11}$ perform the data flow 
analysis by iterating over the ADG until the analysis converges. At an iteration, each node is picked and removed from the work 
list and ${\tt Gen_n}$, ${\tt AGen_n}$, and ${\tt Out_n}$ are computed. If the newly computed ${\tt Out_n}$ differs from its 
old value, the node is pushed back to the work list. The process iterates until the work list becomes empty. 
Post convergence, ${\tt \mathbb{D}^f}$ and ${\tt \mathbb{I}}$ are computed in line number ${\tt 12}$ and the values are returned.

The complexity of the algorithm depends on the number of iterations and the amount of workload per iteration. 
The number of iterations is derived from the maximum depth ${\tt d(G)}$ of the ADG, i.e., the maximum number 
of back edges in any acyclic path derived from the reverse postorder traversal of the graph. 
Therefore, the total number of iterations is ${\tt 1+d(G)+1}$, where the first iteration computes the initial 
values of ${\tt Out_n}$ for all the nodes in the ADG, ${\tt d(G)}$ iterations backpropagate the values of ${\tt Out_n}$, 
and the last iteration verifies the convergence. In the current example, the reverse postorder traversal of the 
ADG produces the acyclic path ${\tt r_{A[i][j]}}\to$ ${\tt r_{X[j]}}\to$ ${\tt w_{Y[i]}}\to$ ${\tt r_{Y[i]}}$, 
containing a single back edge ${\tt r_{Y[i]}}\to$ ${\tt w_{Y[i]}}$. Therefore, the total number of iterations becomes ${\tt 3}$. 

The amount of workload per iteration is dominated by the computation of ${\tt Gen_n}$. In the case of a binary operation, the complexity 
of ${\tt Gen_n}$ is bound to ${\tt\mathcal{O}(d^\prime\times d^{\prime\prime})}$, where ${\tt Out_{p_1}=d^\prime}$, 
${\tt Out_{p_2}=d^{\prime\prime}}$, and ${\tt\{p_1, p_2\}\in Pred_n}$. In the case of assignment and unary operations, 
the complexity of ${\tt Gen_n}$ is bound to ${\tt\mathcal{O}(d^\prime)}$. 

\subsection{Correctness of Essential Indices Analysis}
\label{sec:correctness}
The following claims are sufficient to prove the correctness of our analysis.
\begin{itemize}
	\item\textit{Claim 1}: Every essential data index will always be considered essential.
	\item\textit{Claim 2}: A data index considered essential will not become non-essential later. 
\end{itemize}
	 
Before reasoning about the aforementioned claims we provide an orthogonal lemma to show the correctness of our abstraction.
	\begin{lemma}
	\label{lem:sound}
	Our abstraction function is sound.
	\end{lemma}
	\begin{proof}
	Our abstraction function $\alpha$ maps the concrete value domain of $2^\mathbb{R}$ to the abstract 
		value domain $\{{\tt Z,NZ}\}$. 
	$\{{\tt 0}\}$ in the concrete domain maps to ${\tt Z}$ in the abstract domain and all other elements map to ${\tt NZ}$. 
	Now to guarantee the soundness of $\alpha$ one needs to prove that the following condition~\cite{spa} holds. 
	\begin{equation}
	\label{eq:cond}
		{\tt f(\alpha(c))\sqsubseteq\alpha(cf(c))} 
	\end{equation}
		where ${\tt c}$ is an element in the concrete domain,
		${\tt f}$ is an auxiliary function in the abstract domain, and 
		${\tt cf}$ is the corresponding function in the concrete domain.
		This condition essentially states that the evaluation of function 
		in the abstract domain should overapproximate the evaluation of function in the concrete domain. We prove the above 
		condition for evaluation of each admissible statement in the following lemmas.
	\end{proof}
	\begin{lemma}
		\label{lem:l1}
	For copy assignment statement ${\tt d=d^\prime}$, 
	\begin{center}
		${\tt val(d^\prime)}\sqsubseteq {\tt \alpha(cval(d^\prime))}$.
		\end{center}
	\end{lemma}
	\begin{lemma}
		\label{lem:l2}
	For statement using unary operation ${\tt d=op(d^\prime)}$, 
	\begin{center}
		${\tt op(val(d^\prime))\sqsubseteq \alpha(op(cval(d^\prime)))}$.
		\end{center}
	\end{lemma}
	\begin{lemma}
		\label{lem:l3}
		For statement using binary operation ${\tt d=op(d^\prime,d^{\prime\prime})}$, 
		\begin{center}
		${\tt op(val(d^\prime),val(d^{\prime\prime}))\sqsubseteq \alpha(op(cval(d^\prime),Cali(d^{\prime\prime})))}$.
		\end{center}
	\end{lemma}
	We prove lemmas \ref{lem:l1} to \ref{lem:l3} in the following.
	\begin{proof}
	Let ${\tt cval(d^\prime)=r_1}$ and ${\tt cval(d^{\prime\prime})=r_2}$ where ${\tt r_1}$ and ${\tt r_2}$ are non-zero 
	elements in the concrete domain and ${\tt val(d^\prime)}$ ${\tt=}$ ${\tt val(d^\prime)}$ ${\tt=}$ ${\tt NZ}$ where 
	${\tt NZ}$ represents abstract non-zero value. 
	\begin{figure*}[htbp]
	\begin{center}
		\scalebox{.9}{
	\begin{tabular}{|l|l|l|}
	\hline
	statement & concrete evaluation & abstract evaluation \\
	\hline
	\hline
		${\tt d=d^\prime}$ & ${\tt\alpha(cval(d^\prime))=NZ}$ & ${\tt val(d^\prime)=NZ}$\\
	\hline
		${\tt d=op(d^\prime)}$ & ${\tt \alpha(op(cval(d^\prime)))=NZ}$ & ${\tt op(val(d^\prime))=NZ}$ \\
	\hline
		${\tt d=op(d^\prime, d^{\prime\prime})}$ & ${\tt\alpha(op(cval(d^\prime)}$, ${\tt cval(d^{\prime\prime})))}$ = ${\tt NZ}$ & ${\tt op(val(d^\prime)}$, ${\tt val(d^{\prime\prime}))=NZ}$ \\
	\hline
	\end{tabular}}
	\caption{Concrete and abstract evaluations of statements.}
	\label{fig:table2}
	\end{center}
	\end{figure*}
	From Figure~\ref{fig:table2} we can state that the concrete and abstract evaluations of all statements satisfy the safety condition in 
	Equation~\ref{eq:cond}.
	\end{proof}
	Claim 1 primarily asserts that an essential data index will never be considered non-essential.
	We prove it using induction on the length of paths in the access dependence graph.
	\begin{proof}[Proof of Claim 1]
	Let $\mathbb{D}$ be the set of essential data indices computed at each point in ADG. 
	\begin{itemize}
		\item\textit{Base condition}: At path length ${\tt 0}$, $\mathbb{D}=\mathbb{D}^{\tt 0}$ where $\mathbb{D}^{\tt 0}$ 
	is the initial set of essential data indices of input sparse matrices.
	\item\textit{Inductive step}: Let us assume that at length ${\tt l}$ the set of essential data indices does not miss any 
	essential data index. As abstract computation of such data index is safe as per Lemma~\ref{lem:sound}, 
	we can conclude that no essential data index is missing from $\mathbb{D}$ computed at path length ${\tt l+1}$.
	\end{itemize} 
	Hence all essential data indices will always be considered as essential. 
	\end{proof}
	 
	Because of the monotonicity of transfer functions as the newly generated information is only added to the previously 
	computed information without removing any, we assert that once computed no essential data index will ever be considered 
	as non-essential as stated in Claim 2. 
	 
	For all statements admissible in our analysis the abstraction is optimal except for addition and subtraction operations 
	where numerical cancellation in concrete domain results into ${\tt NZ}$ in the abstract domain. 

\section{Code Generation}
\label{sec:codegen}
In this section, we present the generation of code, customized to the matrix operation and the sparsity structures of input. 
Essential data indices ${\tt \mathbb{D}^f}$ and essential iteration indices ${\tt \mathbb{I}}$ play a crucial role in code 
generation. The fill-in elements generated during the execution alter the structure of the underlying data storage and pose challenges 
in the dynamic alteration of the same. ${\tt\mathbb{D}^f}$ statically identifies the fill-in elements and sets the data storage 
without any requirement for further alteration.

The set of essential iteration indices ${\tt\mathbb{I}}$ identifies the statement instances that are critical for the 
semantic correctness of the operation. In the case of a multi-statement operation, it identifies the essential
statement instances of all the statements present in the loop. The lexicographic ordering of the iteration  
indices statically constructs the execution trace ${\tt E_{trace}}$ of a single statement operation. However, a multi-statement 
operation requires the lexicographic ordering of the timestamp vectors associated with the statement instances, where 
the timestamp vectors identify the order of loops and their nesting sequences. Assuming the ${\mathit{timestamp}}$ function 
computes the timestamp of each essential index and the ${\mathit{lexorder}}$ lexicographically orders the timestamp vectors 
to generate the execution trace ${\tt E_{trace}}$ as follows. 
        \begin{equation}
        \label{eq:etrace}
		{\tt E_{trace} = \mathit{lexorder}\bigl(\bigcup_{\forall e\in I} \mathit{timestamp}(e)\bigl)}
        \end{equation}

\begin{figure}[htbp]
	\begin{center}
		\scalebox{.8}{
		\begin{tabular}{l|l|l|l}
${\tt S_4 : {\langle 0,2\rangle}}$ & $\langle{\tt  S_4,0}\rangle$ & $\langle{\tt S_4}, \langle{\tt A,0,0}\rangle,\langle{\tt A,0,0}\rangle\rangle$ & $\langle{\tt S_4},\langle{\tt valA,0}\rangle,\langle{\tt valA,0}\rangle\rangle$ \\
${\tt S_4 : {\langle 1,2\rangle}}$ & $\langle{\tt S_4,1}\rangle$ & $\langle{\tt S_4}, \langle{\tt A,1,1}\rangle, \langle{\tt A,1,1}\rangle\rangle$ & $\langle{\tt S_4}, \langle{\tt valA,2}\rangle, \langle{\tt valA,2}\rangle\rangle$ \\
${\tt S_4 : {\langle 2,2\rangle}}$ & $\langle{\tt  S_4,2}\rangle$ & $\langle{\tt S_4}, \langle{\tt A,2,2}\rangle, \langle{\tt A,2,2}\rangle\rangle$ & $\langle{\tt S_4}, \langle{\tt valA,4}\rangle, \langle{\tt valA,4}\rangle\rangle$ \\
${\tt S_2 : {\langle 3,0,2,1\rangle}}$ & $\langle{\tt S_2,3,2}\rangle$ & $\langle{\tt S_2}, \langle{\tt A,3,2}\rangle, \langle{\tt A,3,2}\rangle, \langle{\tt A,2,2}\rangle\rangle$ & $\langle{\tt S_2}, \langle{\tt valA,6}\rangle, \langle{\tt valA,6}\rangle, \langle{\tt valA,4}\rangle\rangle$ \\
${\tt S_3 : {\langle 3,1,2\rangle}}$ & $\langle{\tt S_3,3,2}\rangle$ & $\langle{\tt S_3}, \langle{\tt A,3,3}\rangle, \langle{\tt A,3,3}\rangle, \langle{\tt A,3,2}\rangle, \langle{\tt A,3,2}\rangle\rangle$ & $\langle{\tt S_3}, \langle{\tt valA,7}\rangle, \langle{\tt valA,7}\rangle, \langle{\tt valA,6}\rangle, \langle{\tt valA,6}\rangle\rangle$\\
${\tt S_4 : {\langle 3,2\rangle}}$ & $\langle{\tt S_4,3}\rangle$ & $\langle{\tt S_4}, \langle{\tt A,3,3}\rangle, \langle{\tt A,3,3}\rangle\rangle$ & $\langle{\tt S_4}, \langle{\tt valA,7}\rangle, \langle{\tt valA,7}\rangle\rangle$\\
${\tt S_2 : {\langle 4,0,3,1\rangle}}$ & $\langle{\tt S_2,4,3}\rangle$ & $\langle{\tt S_2}, \langle{\tt A,4,3}\rangle, \langle{\tt A,4,3}\rangle, \langle{\tt A,3,3}\rangle\rangle$ & $\langle{\tt S_2}, \langle{\tt valA,10}\rangle, \langle{\tt valA,10}\rangle, \langle{\tt valA,7}\rangle\rangle$\\
${\tt S_3 : {\langle 4,1,3\rangle}}$ & $\langle{\tt S_3,4,3}\rangle$ & $\langle{\tt S_3}, \langle{\tt A,4,4}\rangle, \langle{\tt A,4,4}\rangle, \langle{\tt A,4,3}\rangle, \langle{\tt A,4,3}\rangle\rangle$ & $\langle{\tt S_3}, \langle{\tt valA,11}\rangle, \langle{\tt valA,11}\rangle, \langle{\tt valA,10}\rangle, \langle{\tt valA,10}\rangle\rangle$\\
$\ldots$ & $\ldots$ & $\ldots$ & $\ldots$ \\
\multicolumn{1}{c}{(a)} & \multicolumn{1}{c}{(b)} & \multicolumn{1}{c}{(c)} & \multicolumn{1}{c}{(d)}
		\end{tabular}} 
\caption{Generation of execution trace and memory access trace; (a) Lexicographic ordering of timestamp vectors associated with the statement instances, (b) Execution 
		trace, (c) Data access trace accessing dense storage, (d) Data access trace accessing sparse storage.}
		\label{fig:trace}
	\end{center}
\end{figure}
\begin{example}{\rm
Assuming the timestamp vectors as ${\tt\langle i,0,j,0,k\rangle}$, ${\tt\langle i,0,j,1\rangle}$, ${\tt\langle j,1,l\rangle}$, 
and ${\tt\langle i,2\rangle}$ for the statements ${\tt S_1}$, ${\tt S_2}$, ${\tt S_3}$, and ${\tt S_4}$ in Figure~\ref{fig:motiv1}(a), 
Figures~\ref{fig:trace}(a) and \ref{fig:trace}(b) present the snippets of lexicographic order of the timestamp vector instances 
and the generated execution trace respectively. Here execution instance ${\tt\langle S_k,i,j\rangle}$ denotes the instance of 
statement ${\tt S_k}$ at iteration index ${\tt(i,j)}$.
}\framebox{}
\end{example}

The problem of constructing piecewise regular loops from the execution trace is similar to the problem addressed by 
Rodr{\'i}guez et al.~\cite{piecewise1} and Augustine et al.~\cite{piecewise2}. Their work focuses on homogeneous execution 
traces originating from single statement loops where reordering statement instances is legitimate. They note that handling 
multi-statement loops is out of the scope of their work. They construct polyhedra from the reordered and equidistant execution 
instances and use CLooG~\cite{1342537} like algorithm to generate piecewise-regular loop-based code from the polyhedra. They support 
generating either one-dimensional or multi-dimensional loops. 

Our work targets generic loops including both single-statement and multi-statements, having loop-independent or
loop-dependent dependencies. In the case of multi-statement loops, the instances of different 
statements interleave, affecting the homogeneity of the execution trace. Such interleaving limits the size of 
the homogeneous sections of the trace that contribute to loop generation. Additionally, most 
loops showcase loop-dependent dependencies, and thus, reordering statement instances may affect the semantic 
correctness of the program. Taking these behaviors of programs into account, we use a generic approach to 
generate one-dimensional piecewise regular loops from the homogeneous and equidistant statement instances without 
altering their execution order. 

The execution trace ${\tt E_{trace}}$ prepares the memory access trace ${\tt M_{trace}}$, accessing the underlying 
storage constructed by the essential data indices ${\tt\mathbb{D}^f}$. Assuming $\mathit{memaccess}$ returns the 
data accessed by each iteration index ${\tt e}$ in the execution trace ${\tt E_{trace}}$, ${\tt M_{trace}}$ is computed 
as follows.
\begin{equation}
	\label{eq:memtrace}
	{\tt M_{trace} = \bigcup_{\forall e\in E_{trace}}\mathit{memaccess}(e,\mathbb{D}^f)}
\end{equation}

Instead of a single-dimensional data access trace, i.e., a memory access trace generated by a single operand accessing 
sparse data, our code generation technique considers a multi-dimensional data access trace, where the memory access trace 
is generated by multiple operands accesing sparse data. In the case of a loop statement 
${\tt A[i]} = f({\tt B[j]})$, $\{\ldots,\langle{\tt A,m}\rangle,\ldots,\langle{\tt A,n}\rangle,\ldots\}$ and 
$\{\ldots,\langle{\tt B,m^\prime}\rangle,\ldots,\langle{\tt B,n^\prime}\rangle,\ldots\}$ represent two single-dimensional 
data access traces generated by accessing arrays {\tt A} and {\tt B} respectively. Thus the multi-dimensional 
data access trace generated by the statement is $\{\ldots,\langle\langle{\tt A,m}\rangle, \langle{\tt B,m^\prime}\rangle\rangle$, $\ldots,
\langle\langle{\tt A,n}\rangle$, $\langle{\tt B,n^\prime}\rangle\rangle,\ldots\}$. Note that, if the underlying data storage changes 
the data access trace changes too. 

\begin{example}{\rm
	Figure~\ref{fig:trace}(c) and \ref{fig:trace}(d) represent the snippet of multi-dimensional data access 
	trace accessing dense and sparse storage respectively. Data access point $\langle S_4, \langle{\tt A,0,0}\rangle,
	\langle{\tt A,0,0}\rangle\rangle$ denotes accessing memory location ${\tt A[0][0]}$ of the dense storage by the 
	left-hand side and right-hand side operands of statement ${\tt S_4}$. Similarly, $\langle S_4, 
	\langle{\tt valA,0}\rangle,\langle{\tt valA,0}\rangle\rangle$ represents corresponding accesses to ${\tt valA[0]}$ of 
	the sparse storage. 
	}\framebox{}
\end{example}

The code generator parses the execution trace to identify the homogeneous sections and computes distance vectors between consecutive 
multi-dimensional data access points originated by the same homogeneous section. If data access points ${\tt m_{i-1}}$, ${\tt m_i}$, 
and ${\tt m_{i+1}}$ of ${\tt M_{trace}}$ are homogeneous and equidistant, then they form a partition which is later converted into 
a regular loop. The distance vector between data access points $\langle\langle{\tt A,m}\rangle,\langle{\tt B,m^\prime}\rangle\rangle$ 
and $\langle\langle{\tt A,n}\rangle,\langle{\tt B,n^\prime}\rangle\rangle$ is $\langle\langle{\tt A,n-m}\rangle,
\langle{\tt B,n^\prime-m^\prime}\rangle\rangle$. Homogeneous and equidistant data access points ${\tt\langle A,m\rangle}$, 
${\tt\langle A,m+d\rangle}$, $\ldots$, ${\tt\langle A,m+n\times\!d\rangle}$, with identical distance ${\tt d}$, form an affine, 
one-dimensional, indirect-reference free access function ${\tt A[m+d\times\!i]}$. Iteration index ${\tt i}$ forms a regular 
loop iterating from ${\tt 0}$ to ${\tt n}$. For example, the homogeneous and equidistant data access points 
$\langle{\tt S_4},\langle{\tt valA,0}\rangle,$ $\langle{\tt valA,0}\rangle\rangle$, $\langle{\tt S_4},\langle{\tt valA,2}\rangle,$ 
$\langle{\tt valA,2}\rangle\rangle$, and $\langle{\tt S_4},\langle{\tt valA,4}\rangle,$ $\langle{\tt valA,4}\rangle\rangle$ is 
$\langle\langle{\tt valA,2}\rangle,\langle{\tt valA,2}\rangle\rangle$ construct one dimensional, affine access function 
$\{{\tt valA[2i+0]|0\le\!i\le\!2\}}$.

In the absence of regularity, our technique generates small loops with iteration-size two. 
As this hurts performance because of instruction cache misses, Augustine et al.~\cite{piecewise2} proposed 
instruction prefetching for the program code. However, we deliberately avoid prefetching and reordering in our 
current work and limit the code generation to code that is free of indirect references, 
and contains one-dimensional and piecewise-regular loops for generic programs. 
        \begin{figure}[t]
        \begin{tabular}{c}
        \hline
        \noindent
        \scalebox{.9}{
       	\begin{minipage}[t]{.7\linewidth}
        \begin{algorithm}[H]
        \caption{Algorithm for code generation.}
        \label{algo:codegen}
        \SetAlgoLined
        \LinesNumbered
        \SetKwInOut{Input}{Input}
        \SetKwInOut{Output}{Output}
        \SetKwRepeat{Do}{do}{while}
        \Input{
                \begin{inparaenum}[(a)]
		\item Set of essential data indices ${\tt \mathbb{D}^f}$,
		\item Set of essential iteration indices ${\tt \mathbb{I}}$.
		\end{inparaenum}
	}
	\Output{
		Code ${\tt C}$ containing piecewise-regular loops.
	}
	Compute ${\tt E_{trace}}$ and ${\tt M_{trace}}$ using Equation~\ref{eq:etrace} and Equation~\ref{eq:memtrace} respectively. \\
	\For{${\tt m_i \in M_{trace}}$}{
		\If{${\tt m_{i-1}\in P}$ and ${\tt m_{i-1}, m_i}$ are homogeneous and equidistant}{
		${\tt P = P \cup \{m_i\}}$, ${\tt P}$ be a partition.
		}
		\Else{
		${\tt P^\prime = \{m_i\}}$, ${\tt P^\prime}$ be another partition.
		}
	}
	\For{\textrm{each partition} ${\tt P}$}{${\tt C=C+\mathit{loopgen}(P)}$, $\mathit{loopgen}$ generates affine access function and regular loop}
	\Return ${\tt C}$
	\end{algorithm}
        \end{minipage}} \\
        \hline
        \end{tabular}
        \end{figure}

Algorithm~\ref{algo:codegen} presents the algorithm for code generation. Line number {\tt 1} computes 
${\tt E_{trace}}$ and ${\tt M_{trace}}$. Lines {\tt 2-9} partition ${\tt M_{trace}}$ into multiple partitions, 
containing consecutive, homogeneous, and equidistant data access points. Lines {\tt 10-12} generate regular 
loop for each partition and accumulate them into the code. The complexity of the code generation algorithm 
is bound to the size of the essential iteration indices ${\tt\mathbb{I}}$.

\section{Empirical Evaluation}
\label{sec:result}
\subsection{Experimental Setup}
We have developed a working implementation of SpComp in C++ using STL libraries. It has two modules performing 
the essential indices analysis and piecewise-regular code generation. Our implementation is computation intensive that  
is addressed by parallelizing the high-intensity functions into multiple threads with a fixed workload per thread. 
For our experimentation, we have used Intel(R) Core(TM) i5-10310U CPU @ 1.70GHz octa-core processor with 8GB RAM size, 
4GB available memory size, 4KB memory page size, and L1, L2, and L3 caches of size 256KB, 1MB, and 6MB, respectively. 
The generated code is in {\tt .c} format and is compiled using GCC 9.4.0 with optimization level {\tt -O3} that automatically 
vectorizes the code. Our implementation successfully scales up for Cholesky decomposition to sparse matrix 
{\em Nasa/nasa2146} having ${\tt 7\times10^4}$ non-zero elements but limits the code generation due to the available memory. 

Here we use the PAPI tool~\cite{papi} to profile the dynamic behavior of a code. 
The profiling of a performance counter is performed thousand times, and the mean value is reported. 
The retired instructions, I1 misses, L1 misses, L2 misses, L2I misses, L3 misses, and TLB misses 
are measured using PAPI\_TOT\_INS, ICACHE\_64B : IFTAG\_MISS, MEM\_LOAD\_UOPS\_RETIRED : L1\_MISS, L2\_RQSTS:MISS,
L2\_RQSTS : CODE\_RD\_MISS, LONGEST\_LAT\_CACHE : MISS, and PAPI\_TLB\_DM events respectively. 

\subsection{Use Cases and Experimental Results}
The generated sparsity structure-specific code is usable as long as the sparsity structure remains unchanged. 
Once the structure changes, the structure-specific code no longer remains relevant. In this section, we have 
identified two matrix operations; (a) Sparse Matrix-Sparse Vector Multiplication and (b) Sparse Cholesky decomposition,
that have utility in applications where the sparsity structure-specific codes are reused. 

\subsubsection{Sparse Matrix-Sparse Vector Multiplication}
This sparse matrix operation has utility in applications like page ranking, deep Convolutional Neural Networks (CNN), 
numerical analysis, conjugate gradients computation, etc. Page ranking uses an iterative algorithm that assigns a numerical 
weighting to each vertex in a graph to measure its relative importance. It has a huge application in web page ranking. 
CNN is a neural network that is utilized for classification and computer vision. In the case of CNN training, the sparse 
inputs are filtered by different filters until the performance of CNN converges. 

SpMSpV multiplies a sparse matrix ${\tt A}$ to a sparse vector ${\tt X}$ and outputs a sparse vector ${\tt Y}$. 
It operates on two sparse inputs and generates a sparse output, without affecting the sparsity structure of the inputs. 
The corresponding code operating on dense data contains a perfectly nested loop having a single statement and loop-independent dependencies. 
We compare the performance of SpComp-generated SpMSpV code against the following. 
\begin{itemize}
\item The state-of-art Tensor Algebra Compiler (TACO)~\cite{taco1,taco2} automatically generates the sparse code 
	supporting any storage format. We have selected the storage format of the input matrix ${\tt A}$ as CSR 
	and the storage format of the input vector ${\tt X}$ as a sparse array. The TACO framework~\cite{taco2} 
	does not support sparse array as the output format, thus, we have selected dense array as the output storage format. 
\item The piecewise regular code generated by ~\cite{piecewise1, piecewise2}. We use their working implementation from \textit{PLDI 2019} 
	artifacts~\cite{artifact1} and treat it as a black box. Although, this implementation supports only Sparse Matrix-Vector 
	Multiplication (SpMV) operation, we use this work to showcase the improvement caused by SpComp for multiple sparse 
	input cases. By default, the instruction prefetching is enabled in this framework. However, instruction prefetching 
	raises a {\em NotImplementedError} error during compilation. Thus, we have disabled instruction prefetching for the 
	entire evaluation. 
\end{itemize}
We enable {\tt -O3} optimization level during the compilation of the code generated by TACO, piecewise-regular work, 
and SpComp. Each execution is performed thousand times and the mean is reported. 

The input sparse matrices are randomly selected from the Suitesparse Matrix Collection~\cite{spCollection}, as SpMSpV 
can be applied to any matrix. The input sparse vectors are synthesized from the number of columns of the input sparse 
matrices with sparsity fixed to ${\tt 90\%}$. The initial ${\tt 10\%}$ elements of the sparse vectors are non-zero, 
making the sparsity structured. Such regularity is intentionally maintained to ease the explanation of sparsity structures 
of the input sparse vectors.

\begin{table*}[htbp]
\begin{center}
\scalebox{.8}{
	\begin{tabular}{|cccccc|ccc|}
	\hline
	\multicolumn{6}{|c|}{Sparse matrix} & \multicolumn{3}{c|}{Sparse Vector} \\
	\hline
	Name & Group & Rows & Cols & Nonzeroes & Sparsity & Size & Nonzeroes & Sparsity\\
	\hline
	\hline
	lp\_maros & LPnetlib & ${\tt 846}$ & ${\tt 1966}$ & ${\tt 10137}$ & ${\tt 99.9\%}$ & ${\tt 1966}$ & ${\tt 196}$ & ${\tt 90\%}$\\
	pcb1000 & Meszaros & ${\tt 1565}$ & ${\tt 2820}$ & ${\tt 20463}$ & ${\tt 99.9\%}$ & ${\tt 2820}$ & ${\tt 282}$ & ${\tt 90\%}$\\
	cell1 & Lucifora & ${\tt 7055}$ & ${\tt 7055}$ & ${\tt 30082}$ & ${\tt 99.9\%}$ & ${7055}$ & ${\tt 705}$ & ${\tt 90\%}$\\
	n2c6-b6 & JGD\_Homology & ${\tt 5715}$ & ${\tt 4945}$ & ${\tt 40005}$ & ${\tt 99.9\%}$ & ${\tt 4945}$ & ${\tt 494}$ & ${\tt 90\%}$\\
	beacxc & HB & ${\tt 497}$ & ${\tt 506}$ & ${\tt 50409}$ & ${\tt 99.8\%}$ & ${\tt 506}$ & ${\tt 50}$ & ${\tt 90\%}$\\
	rdist3a & Zitney & ${\tt 2398}$ & ${\tt 2398}$ & ${\tt 61896}$ & ${\tt 99.9\%}$ & ${\tt 2398}$ & ${\tt 239}$ & ${\tt 90\%}$\\
	lp\_wood1p & LPnetlib & ${\tt 244}$ & ${\tt 2595}$ & ${\tt 70216}$ & ${\tt 99.9\%}$ & ${\tt 2595}$ & ${\tt 259}$ & ${\tt 90\%}$\\
	TF15 & JGD\_Forest & ${\tt 6334}$ & ${\tt 7742}$ & ${\tt 80057}$ & ${\tt 99.9\%}$ & ${\tt 7742}$ & ${\tt 774}$ & ${\tt 90\%}$\\
	air03 & Meszaros & ${\tt 124}$ & ${\tt 10757}$ & ${\tt 91028}$ & ${\tt 99.9\%}$ & ${\tt 10757}$ & ${\tt 1075}$ & ${\tt 90\%}$ \\
	Franz8 & JGD\_Franz & ${\tt 16728}$ & ${\tt 7176}$ & ${\tt 100368}$ & ${\tt 99.9\%}$ & ${\tt 7176}$ & ${\tt 717}$ & ${\tt 90\%}$\\
	\hline
	\end{tabular}}
	\caption{Statistics of selected sparse matrices and synthesized sparse vectors for SpMSpV matrix operation.}
	\label{tab:spmspv_stat}
	\end{center}
\end{table*}

The statistics of the selected sparse matrices and synthesized sparse vectors are presented in Table~\ref{tab:spmspv_stat}. 
Due to constraints on the available memory, we limit the number of non-zero elements of the selected sparse matrices between 
${\tt 10000}$ and ${\tt 100000}$. All of the matrices showcase ${\tt \approx 99.9\%}$ sparsity of unstructured nature. 
Only {\em cell1} and {\em rdist3a} sparse matrices are square and the rest of them are rectangular.  

\begin{table*}[htbp]
\begin{center}
\scalebox{.8}{
	\begin{tabular}{|c|cccc|cccc|cccc|}
	\hline
	& \multicolumn{4}{|c|}{TACO} & \multicolumn{4}{|c|}{Piecewise-regular} & \multicolumn{4}{|c|}{SpComp} \\
	\hline
	Name & Rtd & L1 instr & L2 instr & Exec & Rtd & L1 instr & L2 instr & Exec & Rtd & L1 instr & L2 instr & Exec \\
	& instr & miss(\%) & miss(\%) & time(usec) & instr & miss(\%) & miss(\%) & time(usec) & instr & miss(\%) & miss(\%) & time(usec) \\
	\hline
	\hline
lp\_maros & ${\tt 150271}$ & ${\tt 3.8}$ & ${\tt 1.9}$ & ${\tt 64.5}$ & ${\tt 32494}$ & ${\tt 9.9}$ & ${\tt 9.7}$ & ${\tt 33}$ & ${\tt 1134}$ & ${\tt 21.9}$ & ${\tt 15.3}$ & ${\tt 10}$ \\
pcb1000 & ${\tt 230921}$ & ${\tt 2.5}$ & ${\tt 1.3}$ & ${\tt 91}$ & ${\tt 65804}$ & ${\tt 10.4}$ & ${\tt 9.9}$ & ${\tt 80}$ & ${\tt 1447}$ & ${\tt 19.8}$ & ${\tt 14.5}$ & ${\tt 12}$ \\
cell1 & ${\tt 397466}$ & ${\tt 1.4}$ & ${\tt 0.8}$ & ${\tt 135}$ & ${\tt 99444}$ & ${\tt 11.4}$ & ${\tt 11.3}$ & ${\tt 154}$ & ${\tt 14807}$ & ${\tt 12.02}$ & ${\tt 11.4}$ & ${\tt 17}$ \\
n2c6-b6 & ${\tt 417436}$ & ${\tt 1.3}$ & ${\tt 0.7}$ & ${\tt 139}$ & ${\tt 130756}$ & ${\tt 10.8}$ & ${\tt 10.7}$ & ${\tt 188}$ & ${\tt 20970}$ & ${\tt 9.3}$ & ${\tt 8.6}$ & ${\tt 31}$ \\
beacxc & ${\tt 427377}$ & ${\tt 1.3}$ & ${\tt 0.7}$ & ${\tt 147}$ & ${\tt 164709}$ & ${\tt 9.2}$ & ${\tt 9}$ & ${\tt 183}$ & ${\tt 25532}$ & ${\tt 7.02}$ & ${\tt 6.7}$ & ${\tt 30}$ \\
rdist3a & ${\tt 530869}$ & ${\tt 1.1}$ & ${\tt 0.6}$ & ${\tt 165}$ & ${\tt 215732}$ & ${\tt 9.6}$ & ${\tt 9.5}$ & ${\tt 205}$ & ${\tt 40461}$ & ${\tt 1.6}$ & ${\tt 1.4}$ & ${\tt 20}$ \\
lp\_wood1p & ${\tt 562805}$ & ${\tt 1.01}$ & ${\tt 0.5}$ & ${\tt 176}$ & ${\tt 232317}$ & ${\tt 10.8}$ & ${\tt 10.7}$ & ${\tt 342}$ & ${\tt 30820}$ & ${\tt 9.5}$ & ${\tt 9.2}$ & ${\tt 74}$ \\
TF15 & ${\tt 705327}$ & ${\tt 0.8}$ & ${\tt 0.4}$ & ${\tt 233}$ & ${\tt 253965}$ & ${\tt 10.7}$ & ${\tt 10.6}$ & ${\tt 331}$ & ${\tt 23327}$ & ${\tt 11.7}$ & ${\tt 11.4}$ & ${\tt 44}$ \\
air03 & ${\tt 707140}$ & ${\tt 0.8}$ & ${\tt 0.4}$ & ${\tt 242}$ & ${\tt 282753}$ & ${\tt 9.1}$ & ${\tt 8.9}$ & ${\tt 362}$ & ${\tt 40195}$ & ${\tt 6.4}$ & ${\tt 6.2}$ & ${\tt 40}$ \\
Franz8 & ${\tt 972133}$ & ${\tt 0.6}$ & ${\tt 0.3}$ & ${\tt 278}$ & ${\tt 309937}$ & ${\tt 11.4}$ & ${\tt 11.3}$ & ${\tt 483}$ & ${\tt 33049}$ & ${\tt 11.8}$ & ${\tt 8.8}$ & ${\tt 80}$ \\
	\hline				
	\end{tabular}}
	\caption{Performance of the codes generated by TACO, Piecewise-regular, and SpComp for the sparse matrices shown in Table~\ref{tab:spmspv_stat} in terms of number of retired instructions, \% of L1 and L2 instruction misses compared to the retired instructions, and execution time in micro-second(usec). }
	\label{tab:spmspv_perf}
	\end{center}
\end{table*} 

Table~\ref{tab:spmspv_perf} presents the performance achieved by the SpMSpV codes generated by TACO, 
piecewise-regular, and SpComp for the sparse matrices and sparse vectors shown in Table~\ref{tab:spmspv_stat}. 
The performance is captured in terms of the number of retired instructions, \% of retired instructions missed 
by L1 and L2 instruction caches, and execution time in micro-second (usec). We observe significant execution 
time improvement by SpComp compared to both TACO and Piecewise-regular framework. Although SpComp incurs a 
significant amount of relative instruction misses, the major saving happens due to the reduced number of 
retired instructions by the sparsity structure-specific execution of the SpComp-generated code. 

The plot in Figure~\ref{fig:spmspv}(a) illustrates the performance of SpComp compared to TACO. The \% gain 
in execution time is inversely proportional to the \% increment in L1 and L2 instruction misses but is limited 
to the \% reduction of the retired instructions. The increments in relative instruction misses by L1 and L2 
caches occur due to the presence of piecewise-regular loops. Note that, the \% gain and \% reduction by SpComp 
are computed as ${\tt (perf_{taco} - perf_{spcomp}) / perf_{taco} * 100}$, where ${\tt perf_{taco}}$ and 
${\tt perf_{spcomp}}$ denote the performance by TACO and SpComp respectively. Similarly, the \% increment by 
SpComp is computed as ${\tt (perf_{spcomp} - perf_{taco}) / perf_{spcomp} * 100}$.

\begin{figure}[htbp]
	\begin{center}
		\scalebox{.85}{
			\begin{tabular}{cc}
			\includegraphics[width=8cm,height=4.5cm]{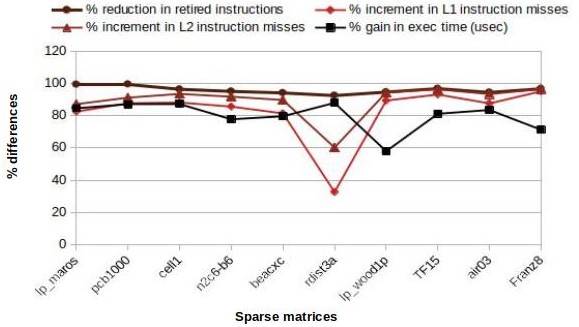}
				&
			\includegraphics[width=8cm,height=4.5cm]{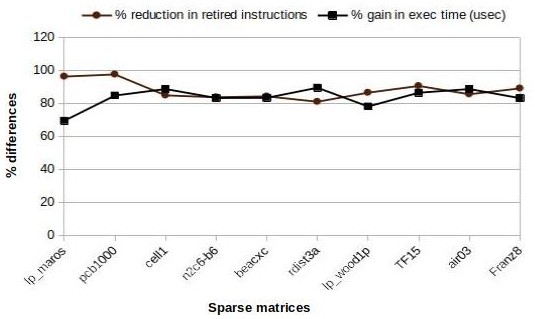} \\
				(a) & (b)
			\end{tabular}}
		\caption{Plots illustrating the performance of SpComp compared to (a) TACO and (b) Piecewise-regular framework.}
		\label{fig:spmspv}
	\end{center}
\end{figure}

As illustrated in the plot in Figure~\ref{fig:spmspv}(b), the \% gain in execution time by SpComp compared to piecewise-regular 
framework is primarily dominated by the \% reduction in retired instructions. This is quite obvious as, unlike the piecewise 
regular work, SpComp considers sparsity of both sparse matrix and sparse vector, making the code specific to both the sparsity 
structures. However, the increments in relative instruction miss by SpComp for {\em lp\_maros} and {\em pcb1000}
occur due to the irregularity present in the SpMSpV output, resulting in piecewise-regular loops of small size. On the contrary, 
SpComp showcases significantly fewer relative instruction misses for {\em rdist3a} as the SpMSpV output showcases high regularity, 
resulting in large-sized loops.

\subsubsection{Sparse Cholesky Decomposition} 
This matrix operation has utility in the circuit simulation domain, where the circuit is simulated until it converges. 
Here the sparsity structure models the physical connections of the circuit which remains unchanged throughout the simulation. 
In each iteration of the simulation the sparse matrix is factorized (Cholesky decomposed in the case of Hermitian 
positive-definite matrices) and the factorized matrix is used to solve the set of linear equations. In this reusable scenario, 
having a Cholesky decomposition customized to the underlying sparsity structure should benefit the overall application performance.

We consider the Cholesky decomposition $\tt chol(A)$, where ${\tt A = LL^*}$ is a factorization of Hermitian positive-definite 
matrix {\tt A} into the product of a lower triangular matrix {\tt L} and its conjugate transpose ${\tt L^*}$.
The operation is mutable, i.e., alters the sparsity structure of the input by introducing fill-in elements, and 
has multiple statements and nested loops with loop-carried and inter-statement dependencies. 

The SpComp-generated code is compared against \textit{CHOLMOD}~\cite{cholmod}, the high-performance library for 
sparse Cholesky decomposition. CHOLMOD applies different ordering methods like \textit{Approximate Minimum Degree}
(AMD)~\cite{10.1145/1024074.1024081}, \textit{Column Approximate Minimum Degree}(COLAMD)~\cite{10.1145/1024074.1024080} 
etc. to reduce the fill-in of the factorized sparse matrix and selects the best-ordered matrix. However, we configure 
both CHOLMOD and SpComp to use only AMD permutation. CHOLMOD offers {\em cholmod\_analyze} and {\em cholmod\_factorize} 
routines to perform symbolic and numeric factorization respectively. We profile the {\em cholmod\_factorize} function 
call for the evaluation. 

We select the sparse matrices from the Suitesparse Matrix Collection~\cite{spCollection}. As the Cholesky decomposition 
applies to symmetric positive definite matrices, it is challenging to identify such matrices from the collection. We have 
noticed that sparse matrices from the structural problem domain are primarily positive definite and thus can be Cholesky 
decomposed. In the collection, we have identified {\tt 200+} such Cholesky factorizable sparse matrices and selected {\tt 35+} 
matrices for our evaluation from the range of {\tt 1000} to {\tt 17000} numbers of nonzero elements. We see that a sparse 
matrix with more nonzeroes exhausts the available memory during code generation and thus is killed.

\begin{table*}
\begin{center}
\scalebox{.8}{
\begin{tabular}{|cccc|cc|cc|}
	\hline
	\multicolumn{4}{|c|}{Input sparse matrix} & \multicolumn{2}{c}{Output sparse matrix} & \multicolumn{2}{|c|}{Generated code} \\
	\hline
	Matrix & Size & Nonzeroes & Sparsity & Nonzeroes & fill-in & Amount of loop in & Avg loop \\
	& & & (\%) & +fill-in & (\%) & generated code(\%) & size \\
	\hline
	\hline
	{nos1} & ${\tt 237\times237}$ & ${\tt 1017}$ & ${\tt 98.19}$ & ${\tt 1094}$ & ${\tt 7.03}$ & ${\tt 37.72}$ & ${\tt 2.35}$ \\
	{mesh3e1} & ${\tt 289\times289}$ & ${\tt 1377}$ & ${\tt 98.35}$ & ${\tt 3045}$ & ${\tt 54.77}$ & ${\tt 85.57}$ & ${\tt 5.57}$ \\
	{bcsstm11} & ${\tt 1473\times1473}$ & ${\tt 1473}$ & ${\tt 99.93}$ & ${\tt 1473}$ & ${\tt 0}$ & ${\tt 100}$ & ${\tt 1473}$ \\
	{can\_229} & ${\tt 229\times229}$ & ${\tt 1777}$ & ${\tt 96.61}$ & ${\tt 3726}$ & ${\tt 52.31}$ & ${\tt 87.64}$ & ${\tt 5.96}$ \\
	{bcsstm26} & ${\tt 1922\times1922}$ & ${\tt 1922}$ & ${\tt 99.95}$ & ${\tt 1922}$ & ${\tt 0}$ & ${\tt 100}$ & ${\tt 1922}$ \\
	{mesh2e1} & ${\tt 306\times306}$ & ${\tt 2018}$ & ${\tt 97.84}$ & ${\tt 4036}$ & ${\tt 50}$ & ${\tt 85.74}$ & ${\tt 5.97}$ \\
	{bcsstk05} & ${\tt 153\times153}$ & ${\tt 2423}$ & ${\tt 89.65}$ & ${\tt 3495}$ & ${\tt 30.67}$ & ${\tt 89.14}$ & ${\tt 7.52}$ \\
	{lund\_b} & ${\tt 147\times147}$ & ${\tt 2441}$ & ${\tt 88.7}$ & ${\tt 3502}$ & ${\tt 30.29}$ & ${\tt 88.49}$ & ${\tt 7.89}$ \\
	{can\_292} & ${\tt 292\times292}$ & ${\tt 2540}$ & ${\tt 97.02}$ & ${\tt 3674}$ & ${\tt 30.86}$ & ${\tt 83.52}$ & ${\tt 4.77}$ \\
	{dwt\_193} & ${\tt 193\times193}$ & ${\tt 3493}$ & ${\tt 90.62}$ & ${\tt 6083}$ & ${\tt 42.57}$ & ${\tt 93.17}$ & ${\tt 9.59}$ \\
	{bcsstk04} & ${\tt 132\times132}$ & ${\tt 3648}$ & ${\tt 79.06}$ & ${\tt 4945}$ & ${\tt 26.22}$ & ${\tt 92.93}$ & ${\tt 9.25}$ \\
	{bcsstk19} & ${\tt 817\times817}$ & ${\tt 6853}$ & ${\tt 98.97}$ & ${\tt 10462}$ & ${\tt 34.49}$ & ${\tt 81.13}$ & ${\tt 4.49}$ \\
	{dwt\_918} & ${\tt 918\times918}$ & ${\tt 7384}$ & ${\tt 99.12}$ & ${\tt 16999}$ & ${\tt 56.56}$ & ${\tt 90.84}$ & ${\tt 8.65}$ \\
	{dwt\_1007} & ${\tt 1007\times1007}$ & ${\tt 8575}$ & ${\tt 99.15}$ & ${\tt 21140}$ & ${\tt 59.43}$ & ${\tt 91.25}$ & ${\tt 8.62}$ \\
	{dwt\_1242} & ${\tt 1242\times1242}$ & ${\tt 10426}$ & ${\tt 99.32}$ & ${\tt 25660}$ & ${\tt 59.37}$ & ${\tt 92.21}$ & ${\tt 9.7}$ \\
	{bcsstm25} & ${\tt 15439\times15439}$ & ${\tt 15439}$ & ${\tt 99.99}$ & ${\tt 15439}$ & ${\tt 0}$ & ${\tt 100}$ & ${\tt 15439}$ \\
	{dwt\_992} & ${\tt 992\times992}$ & ${\tt 16744}$ & ${\tt 98.29}$ & ${\tt 38578}$ & ${\tt 56.59}$ & ${\tt 95.95}$ & ${\tt 8.9}$ \\
	\hline
\end{tabular}}
\caption{Sparsity structures of input and output sparse matrices and statistics of piecewise-regular loops.}
\label{tab:cholmatrices}
\end{center}
\end{table*}

Table~\ref{tab:cholmatrices} presents the sparsity structure of input and output sparse matrices and the structure of 
the generated piecewise regular loops for a few sparse matrices. All the matrices in the table have sparsity within 
the range of {\tt 79\%} to {\tt 99\%} and almost all of them introduce a considerable amount of fill-in when Cholesky 
decomposed. The amount of fill-in(\%) is computed by ${\tt (elem_{out} - elem_{in}) / elem_{out} * 100}$, 
where ${\tt elem_{in}}$ and ${\tt elem_{out}}$ denote the number of non-zero elements before and after factorization. 
Sparse matrices {\em bcsstm11}, {\em bcsstm26}, and {\em bcsstm25} are diagonal, and thus no fill-in 
element is generated when factorized. For these diagonal sparse matrices, SpComp generates a single regular loop with an
average loop size of {\tt 1473}, {\tt 1922}, and {\tt 15439}, the number of non-zero elements. In these cases, 
{\tt 100\%} of the generated code is looped back. 

The rest of the sparse matrices in Table~\ref{tab:cholmatrices} showcase irregular sparsity structures and 
thus produce different amounts of fill-in elements and piecewise regular loops with different average loop sizes. As
an instance, sparse matrix {\tt nos1} with {\tt 98.19\%} sparsity generates {\tt 7.03\%} fill-in elements when Cholesky 
decomposed and {\tt 37.72\%} of generated code is piecewise-regular loops with an average loop size of ${\tt 2.35}$. 
Similarly, another irregular sparse matrix {\tt dwt\_992} with {\tt 98.29\%} sparsity produces {\tt 56.59\%} fill-in 
elements and {\tt 95.95\%} of generated code represents piecewise-regular loops with an average loop size of {\tt 8.9}.

\begin{table*}
	\begin{center}
		\scalebox{.8}{
			\begin{tabular}{|c|ccc|ccc|}
			\hline
			& \multicolumn{3}{|c|}{CHOLMOD} & \multicolumn{3}{c|}{SpComp} \\
			Matrix & Rtd & TLB & Exec & Rtd & TLB & Exec \\
			 & instr & miss & time & instr & miss & time \\
			\hline
			\hline
			nos1 & {\tt 151641} & {\tt 151641} & {\tt 34} &	{\tt 5312} & {\tt 12} & {\tt 7.9} \\
			mesh3e1	& {\tt 425460} & {\tt 425460} & {\tt 108} & {\tt 61960} & {\tt 23} & {\tt 37} \\
			bcsstm11 & {\tt 432790} & {\tt 432790} & {\tt 89} & {\tt 13458} & {\tt 13} & {\tt 6.8} \\
			can\_229 & {\tt 574458}	& {\tt 574460} & {\tt 128} & {\tt 79141} & {\tt 26} & {\tt 42} \\
			bcsstm26 & {\tt 561262} & {\tt 561262} & {\tt 113} & {\tt 17500} & {\tt	29} & {\tt 10} \\
			mesh2e1	& {\tt 585133} & {\tt 585134} & {\tt 144} & {\tt 88186}	& {\tt 27} & {\tt 47} \\
			bcsstk05 & {\tt 492787} & {\tt 492788} & {\tt 106} & {\tt 82730} & {\tt 24} & {\tt 39} \\
			lund\_b	& {\tt 498362} & {\tt 498363} & {\tt 103} & {\tt 80308} & {\tt 25} & {\tt 36} \\
			can\_292 & {\tt 474573} & {\tt 474574} & {\tt 111} & {\tt 59632} & {\tt	31} & {\tt 33} \\
    			dwt\_193 & {\tt 1129067} & {\tt 1129068} & {\tt 217} & {\tt 378388} & {\tt 50} & {\tt 142} \\
			bcsstk04 & {\tt 825630} & {\tt 825630} & {\tt 158} & {\tt 216535} & {\tt 38} & {\tt 87} \\
			bcsstk19 & {\tt 1342154} & {\tt 1342158} & {\tt	344} & {\tt 118177} & {\tt 60} & {\tt 135} \\
			dwt\_918 & {\tt 2843603} & {\tt	2843604} & {\tt	791} & {\tt 763262} & {\tt 96} & {\tt 327} \\
			dwt\_1007 & {\tt 3593429} & {\tt 3593430} & {\tt 888} & {\tt 1007852} & {\tt 139} & {\tt 487} \\
			dwt\_1242 & {\tt 5139767} & {\tt 5139773} & {\tt 1364} & {\tt 1467372} & {\tt 165} & {\tt 596} \\
			bcsstm25 & {\tt	4441657} & {\tt	4441681} & {\tt	1315} & {\tt 139182} & {\tt 83} & {\tt 81} \\
			dwt\_992 & {\tt 8419954} & {\tt 8419954} & {\tt	2123} & {\tt 2775950} & {\tt 276} & {\tt 1270} \\
			\hline
		\end{tabular}}
		\caption{Performance of the codes by CHOLMOD and SpComp for the sparse matrices shown in Table~\ref{tab:cholmatrices} in terms of number of retired 
		instructions, number of TLB miss, and execution time in micro-second (usec).}
		\label{tab:cholperf}
	\end{center}
\end{table*}

\begin{figure}[htbp]
\begin{center}
\scalebox{.9}{
\begin{tabular}{c}
	\includegraphics[width=8.3cm, height=5cm]{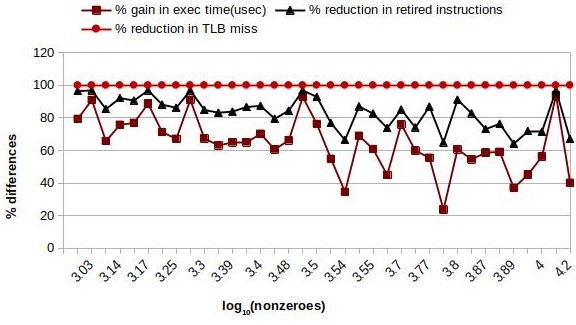}
\end{tabular}}
	\caption{Plot illustrating the performance of SpComp compared to CHOLMOD.}
	\label{fig:graph1}
\end{center}
\end{figure}

The graph in Figure~\ref{fig:graph1} illustrates the performance gained by SpComp against CHOLMOD. The number of 
nonzero elements of sparse matrices is plotted against the logarithmic scale on X-axis. Considering the performance in terms of the
number of retired instructions, the number of TLB miss, and execution time (usec) of CHOLMOD as the baseline, we plot the 
performance difference (in \%) by SpComp against Y-axis. The performance difference is computed as 
${\tt (perf_{cholmod} - perf_{spcomp}) / perf_{cholmod} * 100}$, where ${\tt perf_{cholmod}}$ and ${\tt perf_{spcomp}}$ 
denote the performance by CHOLMOD and SpComp respectively. 

We see a directly proportional relation between \% gain in execution time and \% reduction in the number of retired instructions.  
SpComp contributes to a lesser number of instructions and thus improves the execution time. 
We find $\approx$ {\tt 100\%} reduction in the instructions executed for use cases where the sparse matrices are diagonal, 
like {\em bcsstm25} and {\em bcsstm39}. Additionally, we see $\approx${\tt 100\%} improvement in TLB misses for all the selected 
use cases. This happens due to the static allocation of the fill-in elements that avert the need for dynamic modification of 
sparse data storage, thus improving the TLB miss. Table~\ref{tab:cholperf} presents the raw performance numbers for the sparse matrices.
We see an equal number of retired instructions and TLB misses by CHOLMOD, which implies dynamic memory allocation for all the 
nonzero elements including fill-in elements.

SpComp takes $\approx${\tt 4}sec to perform the analysis on sparse matrix {\tt nos1} and $\approx${\tt 20}min to perform the same 
on sparse matrix {\tt dwt\_992}. As expected, our approach generates large codes even for moderate-sized sparse matrices. In the case of
{\tt dwt\_992} with size ${\tt 992\times 992}$ and NNZ of {\tt 16744} the generated code size is $\approx$ {\tt 6.3} MB. 

\section{Related Work}
\label{sec:relwork}
Here we provide an overview of the work related to optimizing sparse matrix operations either by reorganizing 
data or by reorganizing computation. Over decades researchers have explored various optimization approaches and have established 
various techniques, either hand-crafted or compiler-aided. 
	
Researchers have developed various hand-crafted algorithms involving custom data structures like CSR, CSC, COO, CDS, etc., 
\cite{Eijkhout:1992:LWN:898702,Saad94sparskit:a} that contain the data indices and values of the non-zero data elements. 
Hand-crafted libraries like Cholmod~\cite{cholmod}, Klu~\cite{klu}, CSparse~\cite{csparse} etc. from SuiteSparse~\cite{sp}; 
C++ supported SparseLib++~\cite{spLib,spLib1}, Eigen~\cite{eigen}; Python supported Numpy~\cite{numpy}; 
Intel provided MKL~\cite{mkl}, Radios~\cite{pardiso1,pardiso2}; CUDA supported cuSparse~\cite{cusparse}; Java supported 
Parallel Colt~\cite{pc}; C, Fortran supported PaStiX~\cite{pastix,pastix1}, MUMPS~\cite{mumps}, SuperLU~\cite{superlu} etc. 
are widely used in current practice. Although these libraries offer high-performing sparse matrix operations, they typically require
human effort to build the libraries and port them to different architectures. Also, libraries are often difficult to be 
used in the application and composition of operations encapsulated within separate library functions may be challenging. 

Compiler-aided optimization technique includes run-time optimization approaches like inspection-execution\cite{ins_exe,ins_exe1,ins_exe2} 
where the inspector profiles the memory access information, inspects data dependencies during execution, 
and uses this information to generate an optimized schedule. The executor executes the optimized schedule. 
Such optimization can be even hardware-aware like performing run-time optimization 
for distributed memory architecture~\cite{ins_exe,10.1145/72935.72967,10.1145/1122971.1122990}, and 
shared memory architecture~\cite{RAUCHWERGER1998527,5260538,10.1007/978-3-319-07518-1_8,10.1007/978-3-319-17473-0_9} etc.
Compiler support has been developed to automatically reduce the time and space overhead of inspector-executor
~\cite{Mohammadi:2019:SCD:3314221.3314646,DBLP:journals/corr/abs-1807-10852,Venkat:2015:LDT:2737924.2738003,Venkat:2016:AWP:3014904.3014959,
8436444,Cheshmi:2018:PIT:3291656.3291739,10.1007/3-540-60321-2_3,pnandy}. Polyhedral transformation 
mechanisms~\cite{10.1145/2688500.2688515,10.1145/2581122.2544141,Venkat:2015:LDT:2737924.2738003,
Venkat:2016:AWP:3014904.3014959}, \textit{Sparse Polyhedral Framework} (SPF)~\cite{8436444,STROUT201632} etc. 
address the cost reduction of the inspection. Other run-time approaches~\cite{Kamin:2014:ORS:2775053.2658773,1039784,
Lee:2008:ART:1375527.1375558,10.1007/3-540-58184-7_111,Li2015} propose optimal data distributions during execution such that both 
computation and communication overhead is reduced. Other run-time technique like \textit{Eggs}~\cite{egg} dynamically intercepts the 
arithmetic operations and performs symbolic execution by piggybacking onto Eigen code to accelerate the execution. 
	 
In contrast to run-time mechanisms, compile-time optimization techniques do not incur any execution-time overhead. 
Given the sparse input and code handling dense matrix operation, the work done by \cite{10.1145/181181.181538,
10.1007/3-540-57659-2_4,485501,10.1145/290200.287636,10.5555/645672.665543,bik1993compilation} determine the best storage 
for sparse data and generate the code specific to the underlying storage but not specific to the sparsity structure of 
the input. They handle both single-statement and multi-statement loops and regular loop nests. The generated code contains 
indirect references. These approaches have been implemented in {\em MT1} compiler~\cite{bik1996sparse}, 
creating a sparse compiler to automatically convert a dense program into semantically equivalent sparse code.
Given the best storage for the sparse data, \cite{10.1145/263580.263630,10.5555/269106,10.5555/928197,10.1145/509593.509603,
10.1145/335231.335240,10.5555/866903} propose relational algebra-based techniques to generate efficient sparse matrix 
programs from dense matrix programs and specifications of the sparse input. Similar to \cite{10.1145/181181.181538,
10.1007/3-540-57659-2_4,485501,10.1145/290200.287636,10.5555/645672.665543}, they do not handle mutable cases and generate 
code with indirect references. However, unlike the aforementioned work they handle arbitrary loop nests. Other compile-time 
techniques like \textit{Tensor Algebra Compiler}(TACO)~\cite{taco1,taco2,taco3} automatically generate storage specific code 
for a given matrix operation. They provide a compiler-based technique to generate code for any combination of dense and 
sparse data. \textit{Bernoulli} compiler proposes restructuring compiler to transform a sequential, dense matrix 
Conjugate-Gradient method into a parallel, sparse matrix program~\cite{vladimir}. 

Compared to immutable kernels, compile-time optimization of mutable kernels is intrinsically 
challenging due to the run-time generation of fill-in elements~\cite{10.2307/2156057}. \textit{Symbolic analysis} \cite{davis1} is a 
sparsity structure-specific graph technique to determine the computation pattern of any matrix operation. The information generated by the 
Symbolic analysis guides the optimization of the numeric computation. \cite{sympiler,8665791,cheshmi2022vectorizing} generate 
vectorized and task level parallel codes by decoupling symbolic analysis from the compile-time optimizations. The generated code 
is specific to the sparsity structure of the sparse matrices and is free of indirect references. However, the customization of 
the analysis to handle different kernels requires manual effort. \cite{piecewise1, piecewise2} propose a fundamentally different approach 
where they construct polyhedra from the sparsity structure of the input matrix, and generate indirect reference free regular loops. 
The approach only applies to immutable kernels and the generated code supports out-of-order execution wherever applicable. Their 
work is the closest to our work available in the literature.

Alternate approaches include machine learning techniques and advanced search techniques to select the optimal storage format and 
suitable algorithms for different sparse matrix operations~\cite{Xie:2019:IIA:3330345.3330354,chen2018optimizing,Byun2012AutotuningSM,
10.1007/3-540-57659-2_4,485501}. Apart from generic run-time and compile-time optimization techniques, domain experts have also explored 
domain-specific sparse matrix optimization. As an instance, \cite{wu,Kapre09a,6747987,8252462,6164974,6588795,7160030,7073850,6861585,
7160041,7324551} propose FPGA accelerated sparse matrix operations required in circuit simulation domain. 

\section{Conclusions and Future Work}
\label{sec:conclusion}	 
SpComp is a fully automatic sparsity-structure specific compilation technique that uses data flow analysis to statically
generate piecewise-regular codes customized to the underlying sparsity structures. The generated code is free of indirect 
access and is amenable to SIMD vectorization. It is valid until the sparsity structure changes. 

We focus on the sparsity structure of the output matrices and not just that of the input matrices.
The generality of our method arises from the fact that we drive our analysis by the sparsity structure of the
output matrices which depend on the sparsity structure of the input matrices and hence is covered by the analysis. 
This generality arises from our use of abstract interpretation-based static analysis. Unlike the state-of-art methods, 
our method is fully automatic and does not require any manual effort to customize to different kernels.
	 
In the future, we would like to parallelize our implementation which suffers from significant computation 
overhead and memory limitations while handling large matrices and computation-intensive kernels. We would 
also like to explore the possibility of using GPU-accelerated architectures for our implementation.
	 
Currently, we generate SIMD parallelizable code specific to shared memory architectures.
In the future, we would like to explore the generation of multiple programs multiple data (MPMD)
parallelized codes specific to distributed architectures.
	 
Finally, the current implementation considers the data index of each non-zero element individually.
We would like to explore whether the polyhedra built from the sparsity structure
of the input sparse matrices can be used to construct the precise
sparsity structure of the output sparse matrices.

\bibliographystyle{ACM-Reference-Format}
\bibliography{bibfile}

\end{document}